\documentclass[runningheads]{llncs}
\usepackage{latexsym,amsfonts,amssymb,amsmath}
\usepackage{pdfsync}
\usepackage{tikz-cd} 			

\usepackage{enumerate}			
\usepackage{stmaryrd}			
\usepackage{units} 				

\usepackage{amssymb}			
\usepackage[normalem]{ulem}		
\usepackage{bbding}				
\usepackage{endnotes}			

\usepackage{xcolor}				

\newif\ifBleck

\newif\ifEndNotes 

\newcommand\FnSym{{\scriptsize\PencilLeftDown\kern.1em}} 
\newcommand\EnSym {{$\bigtriangledown$}}

\newif\ifMakeMarkupsCalled 
\newif\ifSuppress 
\newcommand\MakeMarkups[3][.]{ 
 \Suppressfalse 
 \ifBleck\Suppresstrue\fi
 \ifx0#1\Suppresstrue\fi
 \ifx1#1\Suppressfalse\fi
 
 \ifSuppress\expandafter\newcommand\csname#2x\endcsname{\relax}\else
                   \expandafter\newcommand\csname#2x\endcsname{#3}\fi 
 
 \ifSuppress\expandafter\newcommand\csname#2\endcsname[1]{##1}\else
                   \expandafter\newcommand\csname#2\endcsname[1]{{\csname#2x\endcsname##1}}\fi 

 \ifSuppress\expandafter\newcommand\csname#2d\endcsname[1]{{##1}}\else
                   \expandafter\newcommand\csname#2d\endcsname[1]{{\csname#2x\endcsname\sout{##1}}}\fi 

 \ifSuppress\expandafter\newcommand\csname#2r\endcsname[2]{{##1}}\else
                   \expandafter\newcommand\csname#2r\endcsname[2]{\csname#2d\endcsname{##1} \csname#2\endcsname{##2}}\fi 

 \ifSuppress\expandafter\newcommand\csname#2TD\endcsname{\relax}\else
                   \expandafter\newcommand\csname#2TD\endcsname{\csname#2\endcsname{\fbox{\texttt{#2} to do}}}\fi 

 \ifSuppress\expandafter\newcommand\csname#2Bar\endcsname{\relax}\else
                   \expandafter\newcommand\csname#2Bar\endcsname{\csname#2\endcsname{\scriptsize\XSolidBrush}}\fi 

 \ifSuppress\expandafter\newcommand\csname#2f\endcsname[2][]{\relax}\else
  \expandafter\newcommand\csname#2f\endcsname[2][]{
    {\mbox{\csname#2x\endcsname\tiny$\boxtimes$}\marginpar{\csname#2x\endcsname\fbox{\FnSym\footnotemark}}\relax 
    \footnotetext{\csname#2x\endcsname\textsf{##1}~##2}}}\fi
    
 \ifSuppress\expandafter\newcommand\csname#2e\endcsname[1]{\relax}\else
  \expandafter\newcommand\csname#2e\endcsname[1]{
   \EndNotestrue 
   \mbox{\scriptsize\csname#2x\endcsname$\boxtimes$}\relax
   \marginpar{\csname#2x\endcsname\fbox{\EnSym\endnotemark}}
   { 
    \edef\z{~[from~p.\thepage}
    \expandafter\endnotetext\expandafter{\z]~\csname#2x\endcsname ##1\newpage}
   } 
  }\fi
  
 \ifSuppress\expandafter\newcommand\csname#2fe\endcsname[2][]{\relax}\else
  \expandafter\newcommand\csname#2fe\endcsname[2][]{
   \def\File{##1}\relax
   \ifx\File\empty\csname#2f\endcsname{##2}\else 
    \EndNotestrue 
    \mbox{\scriptsize\csname#2x\endcsname$\boxtimes$}\marginpar{\csname#1x\endcsname\fbox{\FnSym\footnotemark}}\relax
    \footnotetext{~\csname#2x\endcsname##2\ --- See Endnote \EnSym\endnotemark\ from file \texttt{##1.tex}.}\relax
    { 
     \edef\z{~[from~Footnote~\thefootnote~on~p.\thepage}
     \expandafter\endnotetext\expandafter{\z]~\csname#2x\endcsname\par\input{##1}\newpage}
    } 
   \fi
  }\fi
 
 \ifSuppress\relax\else
  \csname#2f\endcsname{
   $\backslash$\texttt{#2}$\cdots$\ markups are in \textbf{this} colour\ifx#1..\else\ifx1#1.\else, for #1.\fi\fi
   \ifMakeMarkupsCalled\relax\else 
    \begin{quote}\begin{tabular}{l@{\hspace{2em}}p{.8\linewidth}}
     \multicolumn{2}{l}{\texttt{$\backslash$MakeMarkups\ifx#1.\relax\else[#1]\fi\{#2\}\{{\it$\langle$colour command\/$\rangle$}\}}}\\
         & See comments at \texttt{$\backslash$MakeMarkups} definition. \\[1ex]
     \texttt{$\backslash$#2\{$\langle$text$\rangle$\}} & Sets \texttt{$\langle$text$\rangle$} in \texttt{#2}'s colour. \\
     \texttt{$\backslash$#2x} & Changes to \texttt{#2}'s colour from that point on (until end of context). \\
     \texttt{$\backslash$#2d\{$\langle$text$\rangle$\}} & Sets \texttt{$\langle$text$\rangle$} in \texttt{#2}'s colour with a strikethrough (i.e.\ delete). \\
     \texttt{$\backslash$#2r\{$\langle$this$\rangle$\}\{$\langle$that$\rangle$\}} & Strikes through \texttt{$\langle$this$\rangle$} and inserts \texttt{$\langle$that$\rangle$} (i.e.\ replace). \\
     \texttt{$\backslash$#2f\{$\langle$text$\rangle$\}} & Puts \texttt{$\langle$text$\rangle$} in a \texttt{#2}-coloured footnote with a {\tiny$\boxtimes$} in the main text. \\
     \texttt{$\backslash$#2e\{$\langle$text$\rangle$\}} & Puts \texttt{$\langle$text$\rangle$} in a \texttt{#2}-coloured endnote with a $\boxtimes$ in the main text. \\[.5ex]
     \texttt{$\backslash$#2fe[$\langle$this$\rangle$]\{$\langle$that$\rangle$\}} & Makes a  \texttt{$\backslash$#2f\{$\langle$that$\rangle$\}} that refers to a \\
       & \texttt{$\backslash$#2e\{$\langle$contents of file this.tex$\rangle$\}}. \\ 
       & Without the optional argument, acts as \texttt{$\backslash$#2f\{$\langle$that$\rangle$\}}. \\[.5ex]
     \texttt{$\backslash$#2Bar} & Inserts ``burn after reading'' symbol \csname#2Bar\endcsname, meaning
      \medskip\begin{quote}\begin{itemize}
       \item If yours is the only \csname#2Bar\endcsname\ in this (presumably someone else's) footnote, and you are happy that the footnote has been addressed,
       go ahead and comment-out the whole footnote. (The \csname#2Bar\endcsname\ is their request for you to ``approve and remove''.)
       \item If you are not happy, delete only your \csname#2Bar\endcsname\ and follow-on in the footnote (in your colour, i.e.\ with \texttt{$\backslash$#2x}) saying why you are not happy.
       \item If you are happy, but there are others' burn-after-reading symbols as well as yours, just delete yours; the other people have not yet responded.
      \end{itemize}
      \end{quote}
      \medskip The idea is that when everyone's happy, the last person will comment-out the footnote.
      \\
     \texttt{$\backslash$#2TD} & inserts {\csname#2TD\endcsname}\ . \\
    \end{tabular}\end{quote}
   \fi
  }
  \MakeMarkupsCalledtrue 
 \fi
}

\newcommand\Sec[1] {\S\ref{#1}}
\newcommand\Eqn[1] {(\ref{#1})}
\newcommand\Thm[1] {Thm.~\ref{#1}}
\newcommand\Cor[1] {Cor.~\ref{#1}}
\newcommand\Lem[1] {Lem.~\ref{#1}}
\newcommand\Def[1] {Def.~\ref{#1}}

\newcommand\Fig[1] {Fig.~\ref{#1}}

\renewcommand\Sec[1] {\S\ref{#1}}
\newcommand\App[1] {App.~\ref{#1}}

\newcommand\From {\mbox{\rm :\kern-.05em$\in$\kern.1em}}

\newcommand\Defs {{:=}\,}
\newcommand\Gets {{:}{=}\,}

\newcommand\Dist {{\mathbb D}} 

\newcommand\Wide[1] {~~~#1~~~}

\newcommand\Real {{\mathbb R}}
\newcommand\Fun {\mathbin{\rightarrow}}

\newcommand\In {{:}\,}

\newcommand\Implies {\mathbin{\Rightarrow}}
\newcommand\Uni[1] {\mbox{\small$\Upsilon$}\kern-0.1em_{#1}}

\newcommand{\oset}[2]{%
  {\mathop{#2}\limits^{\vbox to -1.5\ex@{\kern-\tw@\ex@
   \hbox{\scriptsize #1}\vss}}}}
\makeatother

\newcommand\Apply {\mathbin{\raisebox{.1em}{\kern.05em\scriptsize$\rangle$}}}
\newcommand\ApplyR {\mathbin{\raisebox{.1em}{\kern.05em\scriptsize$\langle$}}}

\newcommand\PC[1] {\mathbin{{}_{#1}\kern-.05em\oplus}}

\newenvironment{Reason}{\vspace{-.0em}\begin{tabbing}\hspace{2em}\= \hspace{1cm} \= \kill}
    {\end{tabbing}\vspace{-1em}}
\newcommand\Step[2] {#1 \> $\begin{array}[t]{@{}llll}#2\end{array}$ \\}
\newcommand\StepR[3] {#1 \> $\begin{array}[t]{@{}llll}#3\end{array}$
    \` {\RF \makebox[0pt][r]{\begin{tabular}[t]{r}``#2''\end{tabular}}} \\}
\newcommand\WideStepR[3] {#1 \>
    $\begin{array}[t]{@{}ll}~\\#3\end{array}$ \`
    {\RF \makebox[0pt][r]{\begin{tabular}[t]{r}``#2''\end{tabular}}} \\}
\newcommand\RF {\small}

\usepackage{makeidx}  
\usepackage{tikz}
\usetikzlibrary{arrows, positioning, chains}
\usepackage{subfig}
\usepackage{pgfplots}
\usepackage{times}
\pgfplotsset{compat=1.10}
\usepackage{amssymb, amsmath, amsfonts}
\usepackage{thmtools, thm-restate}
\usepackage{algorithm}
\usepackage{algorithmicx}
\usepackage{algpseudocode}
\usepackage{environ}
\usepackage{hyperref}

\usepackage[square,comma,numbers,sort&compress]{natbib}
\usepackage[]{tocbibind}

\newcommand{\myVec}{\textit{Vec}}
\newcommand{\myDist}{\textit{dist}}
\newcommand{\Lap}[2]{\textit{Lap}^{#1}_{#2}}
\newcommand{\Gam}[2]{\textit{Gam}^{#1}_{#2}}
\newcommand{\Unif}[1]{\textit{Uniform}^{#1}}
\begin{document}

%
%
\pagestyle{headings}  
%

\mainmatter              
\title{Generalised Differential Privacy for Text Document Processing}
%
%
\author{Natasha Fernandes\inst{1,2} \and Mark Dras\inst{1} \and Annabelle McIver\inst{1}
\thanks{We acknowledge the support of the Australian Research Council Grant  DP140101119. }
}
\authorrunning{N. Fernandes et al.} 
\institute{Macquarie University, Sydney, Australia \and
INRIA, Paris-Saclay and \'{E}cole Polytechnique, France}

\maketitle              

\begin{abstract}
We address the problem of how to ``obfuscate'' texts by removing stylistic clues which can identify authorship, whilst preserving (as much as possible) the content of the text.  In this paper we combine ideas from ``generalised differential privacy'' and machine learning techniques for text processing 
 to model privacy for text documents.  We define a privacy mechanism that operates at the level of text documents represented as ``bags-of-words'' --- these representations are typical in machine learning and contain sufficient information to carry out many kinds of classification tasks including \emph{topic identification} and \emph{authorship attribution} (of the original documents). We show that our mechanism satisfies privacy with respect to a metric for semantic similarity, thereby providing a balance between utility, defined by the semantic content of texts, with the obfuscation of stylistic clues.  We demonstrate our implementation on a ``fan fiction'' dataset, confirming that it is indeed possible to disguise writing style effectively whilst preserving enough information and variation for accurate content classification tasks.

\end{abstract}
%

\section{Introduction}

Partial public release of formerly classified data incurs the risk that more information is disclosed than intended.  This is particularly true of data in the form of text such as government documents or patient health records. Nevertheless there are sometimes compelling reasons for declassifying data in some kind of ``sanitised'' form --- for example government documents are frequently released as redacted reports when the law demands it, and health records are often shared to facilitate medical research. Sanitisation is most commonly  carried out by hand but, aside from the cost incurred in time and money, this approach provides no guarantee that the original privacy or security concerns are met. 

To encourage researchers to focus on privacy issues related to text documents the digital forensics community PAN@Clef (\cite{stein:2017l}, for example) proposed a number of challenges that are typically tackled using \emph{machine learning}. In this paper our aim is to demonstrate how to use ideas from \emph{differential privacy} to address some aspects of the PAN@Clef challenges by showing how to provide strong a priori privacy guarantees in document disclosures. 

We focus on the problem of \emph{author obfuscation}, namely to automate the process of changing a given document so that as much as possible  of its original substance remains, but that the author of the document can no longer be identified.  Author obfuscation is very difficult to achieve because it is not clear exactly what to change that would sufficiently mask the author's identity. 
In fact author properties can be determined by ``writing style'' with a high degree of accuracy: this can include author identity \cite{koppel2011authorship} or other undisclosed personal attributes such as native language \cite{tetreault-etal:2013:BEA,malmasi-dras:2018:CL}, gender or age \cite{koppel-etal:2002:LLC,stein:2017o}.  These techniques have been deployed in real world scenarios: native language identification was used as part of the effort to identify the anonymous perpetrators of the 2014 Sony hack \cite{taia:2014}, and it is believed that the US NSA used author attribution techniques to uncover the identity of the real humans behind the fictitious persona of Bitcoin ``creator'' Satoshi Nakamoto.\footnote{\url{https://medium.com/cryptomuse/how-the-nsa-caught-satoshi-nakamoto-868affcef595}}

Our contribution concentrates on the perspective of the ``machine learner'' as an adversary that works with the standard ``bag-of-words'' representation of documents often used in text processing tasks.  A \emph{bag-of-words} representation retains only the original document's words and their frequency (thus forgetting the order in which the words occur). Remarkably this representation still contains sufficient information to enable the original authors to be identified (by a stylistic analysis) \emph{as well as} the document's topic to be classified, both with a significant degree of accuracy.~\footnote{This includes, for example, the character n-gram representation used for author identification in \cite{koppel-winter:2014:JASIST}.}  Within this context we reframe the PAN@Clef author obfuscation challenge as follows:

\begin{quotation}
Given an input bag-of-words representation of a text document, provide a mechanism which changes the input without disturbing its topic classification, but that the author can no longer be identified.
\end{quotation}
 In the rest of the paper we use ideas inspired by \emph{$d_\mathcal X$-privacy} \cite{chatzikokolakis2013broadening}, a metric-based extension of differential privacy, to implement an automated privacy mechanism which, unlike current ad hoc approaches to author obfuscation, gives access to both solid privacy and utility guarantees.\footnote{Our notion of utility here is similar to other work aiming at text privacy, such as \cite{weggenmann2018syntf,li-etal:2018:ACL}.} 
 
 We implement a mechanism $K$ which takes $b, b'$  bag-of-words inputs and produces ``noisy'' bag-of-words  outputs determined by $K(b), K(b')$ with the following properties:

\begin{enumerate}
\item [{\bf Privacy:}] If $b, b'$ are classified to be ``similar in topic'' then, depending on a privacy parameter $\epsilon$ the \emph{outputs} determined by $K(b)$ and $K(b')$ are also ``similar to each other'', irrespective of authorship.
\item[{\bf Utility:}] Possible outputs determined by $K(b)$ are distributed 
according to a Laplace probability density function scored according to a semantic similarity metric.
\end{enumerate}

In what follows we define \emph{semantic similarity} in terms of the classic \emph{Earth Mover's distance} used in machine learning for topic classification in text document processing.~\footnote{In NLP, this distance measure is known as the Word Mover's distance. We use the classic Earth Mover's here for generality.}
We explain how to combine this with $d_\mathcal X$-privacy which extends privacy for databases to other unstructured domains (such as texts). 

 In \Sec{s1501} we set out the details of the bag-of-words representation of documents and define the Earth Mover's metric for topic classification. In \Sec{s1520} we define a generic mechanism which satisfies ``$E_{d_{\mathcal X}}$-privacy'' relative to the Earth Mover's metric $E_{d_\mathcal X}$ and show how to use it for our obfuscation problem. We note that our generic mechanism is of independent interest for other domains where the Earth Mover's metric applies. 
 In \Sec{s1542} we describe how to implement the mechanism for data represented as real-valued vectors and prove its privacy/utility properties with respect to the Earth Mover's metric; in \Sec{s1210} we show how this applies to bags-of-words. 
 Finally in \Sec{s1545} we provide an experimental evaluation of our obfuscation mechanism, and discuss the implications.
 \medskip
 
Throughout we assume standard definitions of probability spaces \cite{Grimmett:92}.  For a set ${\cal A}$ we write $\Dist{\cal A}$ for the set of (possibly continuous) probability distributions over ${\cal A}$. For $\eta \in \Dist{\cal A}$, and $A \subseteq {\cal A}$ a (measurable) subset  we write $\eta(A)$ for the probability that (wrt.\ $\eta$) a  randomly selected $a$ is contained in $A$. In the special case of singleton sets, we write $\eta\{ a\}$. If mechanism $K\In {\alpha}\Fun \Dist\alpha$, we write $K(a)(A)$ for the probability that if the input is $a$, then the output will be contained in $A$.
 

%
%

\section{Documents, topic classification and Earth Moving}\label{s1501}

In this section we summarise the elements from machine learning and text processing  needed for this paper.  Our first definition sets out the representation for documents we shall use throughout. It is a typical representation of text documents used in a variety of classification tasks.

\begin{definition}\label{d1528}
Let ${\cal S}$ be the set of all words (drawn from a finite alphabet). A \emph{document} is defined to be a finite bag over ${\cal S}$, also called a \emph{bag-of-words}.  We denote the set of documents as ${\mathbb B}{\cal S}$, i.e.\  the set of (finite) bags over ${\cal S}$. 
\end{definition}

Once a text is represented as a bag-of-words, depending on the processing task, further representations of the words within the bag are usually required. We shall focus on two important representations: the first is when the task is semantic analysis for eg.\ topic classification, and the second is when the task is author identification. We describe the representation for topic classification in this section, and leave the representation for author identification for \Sec{s1210} and \Sec{s1545}.

\subsection{Word embeddings}

Machine learners can be trained to classify the topic of a document, such as ``health'', ``sport'', ``entertainment''; this notion of topic means that the words within documents will have particular semantic relationships to each other.   There are many ways to do this classification, and in this paper we use a technique that has as a key component ``word embeddings'', which we summarise briefly here.  

A \emph{word embedding} is a real-valued vector representation of words where the precise representation has been experimentally determined by a neural network sensitive to the way words are used in sentences \cite{mikolov2013efficient}. Such embeddings have some interesting properties, but here we only rely on the fact that when the embeddings are compared using a distance determined by a pseudometric%
\footnote{Recall that a pseudometric satisfies both the triangle inequality and symmetry;  but different words could be mapped to the same vector and so ${\myDist_{\myVec}}(w_1, w_2) =0$ no longer implies that $w_1=w_2$.}
 on $\Real^n$, words with similar meanings are found to be close together as word embeddings, and words which are significantly different in meaning are far apart as word embeddings.

\begin{definition}\label{d0925}
An $n$-dimensional word embedding is a mapping $\myVec: {\cal S}\Fun \Real^n$.   Given a pseudometric  $\myDist$ on $\Real^n$ we define a distance on words $ {\myDist_{\myVec}} :{\cal S}{\times}{\cal S}{\rightarrow}\Real_{\geq}$ as follows:
$$
 {\myDist_{\myVec}}(w_1, w_2) \Wide{\Defs} \myDist(\myVec(w_1), \myVec(w_2))~.
$$
\end{definition}

Observe that the property of a pseudometric  on $\Real^n$ carries over to ${\cal S}$.  

\begin{lemma}\label{l1227}
If $\myDist$ is a pseudometric on $\Real^n$ then ${\myDist_\myVec}$ is also a pseudometric on ${\cal S}$. 
\begin{proof}
Immediate from the definition of a pseudometric: i.e.\ the triangle equality and the symmetry of ${\myDist_\myVec}$ are inherited from $\myDist$.
\end{proof}
\end{lemma}

Word embeddings are particularly suited to language analysis tasks, including topic classification, due to their useful semantic properties. Their effectiveness depends on the quality of the embedding $\myVec$, which can vary depending on the size and quality of the training data.
We provide more details of the particular embeddings in \Sec{s1545}.
Topic classifiers can also differ on the choice of underlying metric $\myDist$, and we discuss variations in \Sec{s0949}.

In addition, once the word embedding $\myVec$ has been determined, and the distance $\myDist$ has been selected for comparing ``word meanings'', there are a variety of semantic similarity measures that can be used to compare documents, for us bags-of-words. In this work we use the ``Word Mover's Distance'', which was shown to perform well across multiple text classification tasks \cite{kusner2015word}.


The \emph{Word Mover's Distance} is based on the classic \emph{Earth Mover's Distance} \citep{rubner2000earth} used in transportation problems with a given distance measure.  
We shall use the more general Earth Mover's definition with $\myDist$
\footnote{In our experiments we take $\myDist$ to be defined by the Euclidean distance.}
as the underlying distance measure between words. We note that our results can be applied to problems outside of the text processing domain.

%

Let  $X, Y \in \mathbb{B}\cal{S}$; we denote by $X$ the tuple $\langle x_1^{a_1}, x_2^{a_2}, \dots, x_k^{a_k} \rangle$, where $a_i$ is the number of times that  $x_i$ occurs in $X$. Similarly we write $Y = \langle y_1^{b_1}, y_2^{b_2}, \dots, y_l^{b_l} \rangle$; we have $\sum_i a_i = |X|$ and $\sum_j b_j = |Y|$, the sizes of $X$ and $Y$ respectively. We define a \emph{flow matrix} $F \in \mathbb{R}^{k \times l}_{\geq 0}$ where $F_{ij}$ represents the (non-negative) amount of flow from $x_i \in X$ to $y_j \in Y$. 

\begin{definition}{(Earth Mover's Distance)}\label{d0938}
\label{emd}
Let $d_\mathcal{S}$ be a (pseudo)metric over $\cal{S}$.  The Earth Mover's Distance with respect to $d_\mathcal{S}$, denoted by ${E_{d_\mathcal{S}}}$, is the solution to the following linear optimisation:
\begin{align}
    &   {E_{d_\mathcal{S}}}(X, Y) \Wide{\Defs}  \min \sum\limits_{x_i \in X} \sum\limits_{y_j \in Y} d_\mathcal{S}(x_i, y_j) F_{ij} \label{min1}~,  \quad \textit{subject to:}\\
       & \sum\limits_{i = 1}^{k} F_{ij} = \frac{b_j}{|Y|}  \label{const1} \quad \textit{and}\quad
        \sum\limits_{j = 1}^{l} F_{ij} = \frac{a_i}{|X|}~, \quad  F_{ij} \ge 0, \quad1 \le i \le k, 1 \le j \le l   
\end{align}
where the minimum in (\ref{min1}) is over all possible flow matrices $F$ subject to the constraints (\ref{const1}). In the special case that $|X| = |Y|$, the solution is known to satisfy the conditions of a (pseudo)metric~\citep{rubner2000earth} which we call the Earth Mover's Metric.
\end{definition}

In this paper we are interested in the special case $|X| = |Y|$, hence we use the term \emph{Earth Mover's metric} to refer to ${E_{d_\mathcal{S}}}$.

We end this section by  describing how  texts are prepared for machine learning tasks, and how \Def{d0938} is used to distinguish documents. Consider the text snippet ``The President greets the press in Chicago". The first thing is to remove all ``stopwords'' -- these are  words which do not contribute to  semantics, and include things like  prepositions, pronouns and articles. The words remaining are those that contain a great deal of semantic and stylistic traits.%
\footnote{In fact the way that stopwords are used in texts turn out to be characteristic features of authorship. Here we follow standard practice in natural language processing to remove them for efficiency purposes and study the privacy of what remains. All of our results apply equally well had we left stopwords in place.}

In this case we obtain the bag: 
\[
    b_1 \Wide{\Defs} \langle \text{President}^1, ~\text{greets}^1, ~\text{press}^1, ~\text{Chicago}^1\rangle~.
\]
%
%
Consider a second bag: 
$ b_2 {\Defs}  \langle \text{Chief}^1, \text{speaks}^1, \text{media}^1, \text{Illinois}^1\rangle 
$, corresponding to a different text. \Fig{emd-fig1} illustrates the optimal flow matrix which solves the optimisation problem in \Def{d0938} relative to $d_{\cal S}$.  Here each word is mapped completely to another word, so that $F_{i,j}= 1/4$ when $i=j$ and $0$ otherwise. We show later that this is always the case between bags of the same size. 
With these choices we can compute the distance between $b_1, b_2$:
\begin{align}\label{e1023}
    E_{d_{\cal S}}(b_1, b_2)& = \frac{1}{4} (d_{\cal S}(\text{President}, \text{Chief}) + d_{\cal S}(\text{greets}, \text{speaks}) +\notag \\
                                       & \qquad d_{\cal S}(\text{press}, \text{media}) + d_{\cal S}(\text{Chicago}, \text{Illinois}) ) \\
                                      & = 2.816~.\notag
\end{align} 

For comparison, consider  the distance between $b_1$ and $b_2$ to a third document, 
$
    b_3 {\Defs} {\langle \text{Chef}^1, \text{breaks}^1, \text{cooking}^1, \text{record}^1\rangle}.
$
Using the same word embedding metric,~\footnote{We use the same word2vec-based metric as per our experiments; this is described in \Sec{s1545}.} we find that ${E_{d_{\cal S}}}(b_1, b_3) = 4.121$ and ${E_{d_{\cal S}}}(b_2, b_3) = 3.941$. Thus $b_1, b_2$ would be classified as semantically ``closer'' to each other than to $b_3$, in line with our own  (linguistic) interpretation of the original texts.

\begin{figure}[htbp]
\small
\centering
\begin{tikzpicture}[>=stealth,
       baseline,anchor=north,
        every node/.style={block},
        block/.style={minimum height=1em,minimum width=1em,outer sep=3pt,rectangle,node distance=0pt}]
        \node (T01) []  at (6,0) {\textbf{Bag} $b_1$};
        \node (T11) [below=of T01, text width=1.3cm] {President,};
        \node (T21) [below=of T11, text width=1.3cm]  {greets, };
        \node (T31) [below=of T21, text width=1.3cm]  {press, };
        \node (T41) [below=of T31, text width=1.3cm]  {Chicago};    
        
         \draw [decorate, decoration={brace,raise=3pt}] (T11.north east) -- (T41.south east) node (a) [midway,draw=none] {};
         \draw [decorate, decoration={brace,mirror,raise=3pt}] (T11.north west) -- (T41.south west) node (b) [midway,draw=none] {};            
\end{tikzpicture}%
\hskip 0.8cm
\begin{tikzpicture}[>=stealth,
     baseline=(current bounding box.north),
     orig/.style={circle,draw=blue,fill=blue,inner sep=0pt,minimum size=2mm},
     obf/.style={circle,draw,fill=black,inner sep=0pt,minimum size=2mm}]
  \draw [very thick,->] (0,0) -- (6,0) node[below,midway] {};
  \draw [very thick,->] (0,0) -- (0,5) node[left] {};
  \node (x1)  [orig, label={[blue]above:Chief}] at (1,4.1) {};
  \node (x4) [orig, label={[blue]below:Illinois}] at (1.6, 1.1) {};
  \node (x2) [orig, label={[blue]below:speaks}] at (4, 3.8) {};
  \node (x3) [orig, label={[blue]above:media}] at (4.8, 1.5) {};
 
  \node (z1) [obf, label=below:President] at (1.8,3.3) {};
  \node (z2) [obf, label=above:greets] at (5, 4.4) {};
  \node (z3) [obf, label=below:press] at (3.7, 1.3) {};
  \node (z4) [obf, label=Chicago] at (2.4, 1.8) {};

   \path[thick, ->] (x1) edge [right,midway] node {$d_1$} (z1);
   \path[thick, ->] (x2) edge [below=5pt,midway] node {$d_2$} (z2);
   \path[thick, ->] (x3) edge [below=5pt,midway] node {$d_3$} (z3);
   \path[thick, ->] (x4) edge [right,midway] node {$d_4$} (z4);
\end{tikzpicture}%
\hskip 0.8cm
\begin{tikzpicture}[>=stealth,
        baseline=(current bounding box.north),
        every node/.style={block},
        block/.style={minimum height=1em,minimum width=1em,outer sep=3pt,rectangle,node distance=0pt}]
        \node (T01) []  at (6,0) {\textbf{Bag}  $b_2$};
        \node (T11) [below=of T01, text width=1.0cm] {Chief,};
        \node (T21) [below=of T11, text width=1.0cm]  {speaks, };
        \node (T31) [below=of T21, text width=1.0cm]  {media, };
        \node (T41) [below=of T31, text width=1.0cm]  {Illinois};    
        
         \draw [decorate, decoration={brace,raise=3pt}] (T11.north east) -- (T41.south east) node (a) [midway,draw=none] {};
         \draw [decorate, decoration={brace,mirror,raise=3pt}] (T11.north west) -- (T41.south west) node (b) [midway,draw=none] {};            
\end{tikzpicture}%
\caption{\small Earth Mover's metric between sample documents.} 
\label{emd-fig1}
\end{figure}
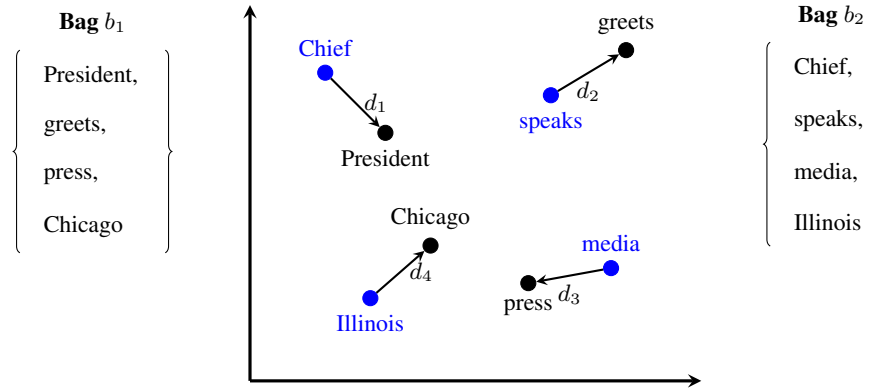

\section{Differential Privacy and the Earth Mover's Metric}\label{s1520}
Differential Privacy was originally defined with the protection of individuals' data in mind.  The intuition is that privacy is achieved through  ``plausible deniability'', i.e.\  whatever output is obtained from a query, it could have just as easily have arisen from  a database that does not contain an individual's details, as from one that does.  In particular, there should be no easy way to distinguish between the two possibilities.  Privacy in text processing means something a little different. A ``query'' corresponds to releasing the topic-related contents of the document (in our case the bag-of-words) --- this relates to the utility because we would like to reveal the semantic content. The privacy relates to investing \emph{individual documents} with plausible deniability, rather than \emph{individual authors} directly. What this means for privacy is the following.  Suppose we are given two documents $b_1, b_2$ written by two distinct authors $A_1, A_2$, and suppose further that $b_1, b_2$ are changed through a privacy mechanism so that it is difficult or impossible to distinguish between them (by any means). Then it is also difficult or impossible to determine whether the authors of the original documents are $A_1$ or $A_2$, or some other author entirely.  This is our aim for obfuscating authorship whilst preserving semantic content.

\medskip

Our approach to obfuscating documents replaces words with other words, governed by probability distributions over possible replacements.  Thus the type of our mechanism is $\mathbb{B}\mathcal{S}\Fun \Dist(\mathbb{B}\mathcal{S})$, where (recall) $\Dist(\mathbb{B}\mathcal{S})$ is the set of probability distributions over the set of (finite) bags of ${\cal S}$. Since we are aiming to find a careful trade-off between utility and privacy, our objective is to ensure that there is a high probability of outputting a document with a similar topic as the input document.
As explained in \Sec{s1501}, topic similarity of documents is determined by the Earth Mover's distance relative to a given (pseudo)metric on word embeddings,
and so our privacy definition must also be relative to the Earth Mover's distance. 

\begin{definition}{(Earth Mover's Privacy)}\label{d1857}
Let $\mathcal{X}$ be a set, and $d_\mathcal{X}$ be a (pseudo)metric on $\mathcal{X}$ and let ${E_{d_\mathcal{X}}}$ be the Earth Mover's metric on $\mathbb{B}\mathcal{X}$ relative to $d_{\cal X}$.
Given $\epsilon \geq 0$, a mechanism $K: \mathbb{B}\mathcal{X} \rightarrow \mathbb{D}(\mathbb{B}\mathcal{X})$ satisfies $\epsilon {E_{d_\mathcal{X}}}$-privacy iff for any $b, b' \in \mathbb{B}\mathcal{X}$ and $Z \subseteq \mathbb{B}\mathcal{X}$:
\begin{equation}\label{e1116} 
      K(b)(Z) \Wide{\le} e^{\epsilon {{E_{d_\mathcal{X}}}(b, b')}} K(b')(Z)~.
\end{equation}
\end{definition}
\Def{d1857} tells us that when two documents are measured to be very close, so that $\epsilon {{E_{d_\mathcal{X}}}(b, b')}$ is close to $0$, then  the multiplier $e^{\epsilon {{E_{d_\mathcal{X}}}(b, b')}}$ is approximately $1$ and the outputs $K(b)$ and $K(b')$ are almost identical. On the other hand the more that the input bags can be distinguished by ${E_{d_\mathcal{X}}}$, the more their outputs are likely to differ. This flexibility is what allows us to strike a balance between utility and privacy; we discuss this issue further in \Sec{s1210} below.

Our next task is to show how to implement a mechanism that can be proved to satisfy \Def{d1857}. We follow the basic construction of Dwork et al. \cite{dwork2006calibrating} for lifting a differentially private mechanism $K \In {\cal X} \Fun \Dist {\cal X}$ to a differentially private mechanism $\underline{K}^\star \In {\cal X}^N \Fun \Dist {\cal X}^N$ on \emph{vectors} in ${\cal X}^N$. (Note that, unlike a bag, a vector imposes a fixed order on its components.)  Here the idea is to apply $K$ independently to each component of a vector $v \in {\cal X}^N$ to produce a random output vector, also in ${\cal X}^N$.  In particular the probability of outputting some vector $v'$ is the product:

\begin{equation}\label{e0857}
\underline{K}^\star(v)\{v'\} \Wide{=} \prod_{1\leq i\leq N} K(v_i)\{v'_i\}~.
\end{equation}
Thanks to the compositional properties of differential privacy when the underlying metric on ${\cal X}$ satisfies the triangle inequality, it's possible to show that the resulting mechanism $\underline{K}^\star$ satisfies the following privacy mechanism \cite{dwork2014}:

\begin{equation}\label{e1939}
\underline{K}^\star(v)(Z) \Wide{\leq} e^{{M_{d_{\cal X}}}(v, v')}\underline{K}^\star(v')(Z) ~,
\end{equation}
where ${M_{d_{\cal X}}}(v, v') \Defs \sum_{1\leq i\leq N}d_{\cal X}(v_i,  v'_i)$, the Manhattan metric relative to $d_{\cal X}$. 

\medskip

However \Def{d1857} does not follow from \Eqn{e1939}, since \Def{d1857} operates on bags of size $N$, and the Manhattan distance between any vector representation of bags is \emph{greater} than $N \times {E_{d_{\cal X}}}$. Remarkably however, it turns out that $K^\star$ --the mechanism that applies $K$ independently to each item in a given bag--  in fact satisfies the much stronger \Def{d1857}, as the following theorem shows, provided the input bags have the same size as each other. 



\newcommand{\restrict}{\!\!\!\downarrow\! N}
\newcommand{\Earth}[1]{E_{#1}}
\newcommand{\Manhattan}[1]{M_{#1}}

\begin{theorem}\label{emprivacy}
Let $d_{\mathcal{X}}$ be a pseudo-metric on $\mathcal{X}$ and let $K: \mathcal{X} \rightarrow \mathbb{D}\mathcal{X}$ be a mechanism satisfying $\epsilon d_{\cal X}$-privacy, i.e.
\begin{equation}\label{e0848}
K(x)(Z) \Wide{\leq} e^{\epsilon d_{\cal X}(x, x')}K(x')(Z)~,~ \textit{for all~} x, x'\in {\cal X},~Z \subseteq {\cal X}.
\end{equation}

 Let $K^\star: \mathbb{B}\mathcal{X} \Fun \Dist( \mathbb{B}\mathcal{X})$ be the mechanism obtained by applying $K$ independently to each element of $X$ for any $X \in \mathbb{B}\mathcal{X}$. Denote by $K^\star\restrict$ the restriction of $K^\star$ to bags of fixed size $N$.  Then $K^\star\restrict$ satisfies $\epsilon N {E_{d_\mathcal{X}}}$-privacy.
%
\begin{proof} (Sketch)
The full proof is given in \App{A0768}; here we sketch the main ideas.

 Let $b, b'$ be input bags, both of size $N$, and let $c$ a possible output bag (of $K^\star$). Observe that both output bags determined by $K^\star(b_1), K^\star(b_2)$ and  $c$ also have size $N$. We shall show that \Eqn{e1116} is satisfied for the set containing the singleton element $c$ and multiplier $\epsilon N$, from which it follows that \Eqn{e1116} is satisfied for all sets $Z$. 
 
By  Birkhoff-von Neumann's theorem (\cite{Konig:36}, \Thm{t1244}), in the case where all bags have the same size,  the minimisation problem in \Def{d0938} is optimised for transportation matrix $F$ where  all values $F_{ij}$ are either $0$ or $1/N$. This implies that  the optimal transportation for $\Earth{d_{\cal X}}(b, c)$ is achieved by moving each word in the bag $b$ to a (single) word in bag $c$. The same is true for $\Earth{d_{\cal X}}(b', c)$ and $\Earth{d_{\cal X}}(b, b')$.  Next we use a vector representation of bags as follows. For bag $b$, we write $\underline{b}$ for a vector in ${\cal X}^N$ such that each element in $b$ appears at some  $\underline{b}_i$.

Next we fix $\underline{b}$ and $\underline{b'}$ to be vector representations  of respectively $b, b'$ in ${\cal X}^N$ such that the optimal transportation for $\Earth{d_{\cal X}}(b, b')$ is 
\begin{equation}\label{e1249}
\Earth{d_{\cal X}}(b, b') \Wide{=}1/N {\times}\sum_{1\leq i\leq N} d_{\cal X}(\underline{b}_i, \underline{b}'_i) \Wide{=} \Manhattan{d_{\cal X}}(\underline{b}, \underline{b}')/N~.
\end{equation}

The final fact we need is to note that there is a relationship between $K^\star$ acting on bags of size $N$ and $\underline{K}^\star$ which acts on vectors in ${\cal X}^N$ by applying $K$ independently to each component of a vector: it is characterised in the following way. Let $b,  c$ be bags and let $\underline{b}, \underline{c}$ be any vector representations. For permutation $\sigma \in \{1 \dots N\}\Fun \{1\dots N\}$ write $\underline{c}^\sigma$ to be the vector with components permuted by $\sigma$, so that $\underline{c}^\sigma_i = \underline{c}_{\sigma(i)}$. With these definitions, the following equality between probabilities holds:
\begin{equation}\label{e1229}
K^\star(b)\{c\} \Wide{=} \sum_{\sigma}\underline{K}^\star(\underline{b})\{\underline{c}^\sigma\}~,
\end{equation}
where the summation is over all permutations that give distinct vector representations of $c$. We now compute directly:

\begin{Reason}
\Step{}
{K^\star(b)\{c\}}
\StepR{$=$}{\Eqn{e1229} for $b, c$}
{\sum_{\sigma}\underline{K}^\star(\underline{b})\{\underline{c}^\sigma\}}
\StepR{$\leq$}{\Eqn{e1939} for $\underline{b}, \underline{b'}, \underline{c}$}
{\sum_{\sigma}e^{\epsilon\Manhattan{d}(\underline{b}, \underline{b}')}\underline{K}^\star(\underline{b}')\{\underline{c}^\sigma\}}
\StepR{$=$}{Arithmetic and \Eqn{e1249}}
{e^{\epsilon N\Earth{d}(\underline{b}, \underline{b}')}{\sum_{\sigma}\underline{K}^\star(\underline{b}')\{\underline{c}^\sigma\}}}
\StepR{$=$}{\Eqn{e1229} for $b', c$}
{e^{\epsilon N\Earth{d}(\underline{b}, \underline{b}')}K^\star(b')\{c\}~,}
\end{Reason}
as required.


%
%
%
\end{proof}
\end{theorem}

\subsection{Application to Text Documents}

Recall the bag-of-words 
\[
              b_2 \Wide{\Defs}  \langle \text{Chief}^1, \text{speaks}^1, \text{media}^1, \text{Illinois}^1\rangle ~,
\]
and assume we are provided with a mechanism $K$ satisfying the standard $\epsilon d_\mathcal{X}$-privacy property \Eqn{e0848} for individual words. As in \Thm{emprivacy} we can create a mechanism $K^*$ by applying $K$ independently to each word in the bag,  so that, for example the probability of outputting  $b_3 = {\langle \text{Chef}^1, \text{breaks}^1, \text{cooking}^1, \text{record}^1\rangle}$  is determined by \Eqn{e1229}:
\[
K^\star(b_2)(\{b_3\}) \Wide{=} \sum_{\sigma}\prod_{1\leq i\leq 4}K({b_2}_i)\{{{\underline{b}_3}_i}^\sigma\}~.
\]

By \Thm{emprivacy},  $K^\star$ satisfies $4\epsilon E_{d_{\cal S}}$-privacy. Recalling \Eqn{e1023} that  $E_{d_{\cal S}}(b_1, b_2) = 2.816$, we deduce that
if $\epsilon \sim 1/16$ then the output distributions  $K^\star(b_1)$ and $K^\star(b_2)$ would differ by the multiplier $e^{2.816{\times 4}/16}\sim 2.02$; but if  $\epsilon \sim 1/32$ those distributions differ by only $1.42$. In the latter case it means that the outputs of $K^\star$ on $b_1$ and $b_2$ are almost indistinguishable. 

The parameter $\epsilon$ depends on the randomness implemented in the basic mechanism $K$; we investigate that further in \Sec{s1542}.

\subsection{Properties of Earth Mover's Privacy}\label{s0949}

\newcommand{\Euclidean}[1]{|\!| #1|\!|}
\newcommand{\newManhattan}[1]{\lfloor #1\rfloor}
\newcommand{\myCosine}[2]{\wr #1 - #2 \wr}

In machine learning a number of ``distance measures'' are used in classification or clustering tasks, and in this section we explore some properties of privacy when we vary the underlying metrics of an Earth Mover's metric used to classify complex objects.

Let $v, v' \in \Real^n$ be real-valued $n$-dimensional vectors. We use the following (well-known) metrics. Recall in our applications we have looked at bags-of-words, where the words themselves are represented as $n$-dimensional vectors.
\footnote{As we shall see, in the machine learning analysis \emph{documents} are represented as bags of $n$-dimensional vectors (word embeddings), where each bag contains $N$ such vectors.}
\begin{enumerate}
\item Euclidean: \quad $\Euclidean{v{-}v'} \Wide{\Defs} \sqrt {\sum_{1\leq i \leq n}(v_i - v'_i)^2}$\\
\item Manhattan: \quad $\newManhattan{v{-}v'} \Wide{\Defs} {\sum_{1\leq i \leq n}|v_i - v'_i|}$\\
\end{enumerate}
%
%
%
%
Note that the Euclidean and Manhattan distances determine pseudometrics on words as defined at \Def{d0925} and proved at \Lem{l1227}.

%

\begin{lemma}{}
If $d_\mathcal{X} \le d_\mathcal{X'}$ (point-wise), then $E_{d_\mathcal{X}} \le E_{d_\mathcal{X'}}$ (point-wise).
\label{lem2}

\begin{proof}
Trivial, by contradiction. If $d_\mathcal{X} \le d_\mathcal{X'}$ and $F_{ij}, F^\star_{ij}$ are the minimal flow matrices for $E_{d_\mathcal{X}},  E_{d_\mathcal{X'}}$ respectively, then $F^\star_{ij}$ is a (strictly smaller) minimal solution for $E_{d_\mathcal{X}}$ which contradicts the minimality of $F_{ij}$.
\end{proof}
\end{lemma}

\begin{corollary}{}
If $d_\mathcal{X} \le d_\mathcal{X'}$ (point-wise), then $E_{d_\mathcal{X}}$-privacy implies $E_{d_\mathcal{X'}}$-privacy.
\label{cor1}
\end{corollary}

This shows that, for example, $E_{\Euclidean{\cdot}}$-privacy implies  $E_{\newManhattan{\cdot}}$-privacy, and indeed any distance measure $d$ which exceeds the Euclidean distance then  $E_{\Euclidean{\cdot}}$-privacy implies  $E_{d}$-privacy. 

\medskip

We end this section by noting that \Def{d1857} satisfies  \emph{post-processing}; i.e. that privacy does not decrease under post processing.
We write $K ; K'$ for the composition of mechanisms $K, K': \mathbb{B}\mathcal{X} \rightarrow \mathbb{D}(\mathbb{B}\mathcal{X})$, defined:

\begin{equation}\label{d1143}
(K ; K')(b)(Z) \Wide{\Defs} \sum_{b' \In  \mathbb{B}\mathcal{X}} K(b)(\{ b'\}){\times} K'(b')(Z)~.
\end{equation}

\begin{restatable}{lemma}{postlemma}
\label{l1112}
[Post processing]  If $K, K' \In \mathbb{B}\mathcal{X} \rightarrow \mathbb{D}(\mathbb{B}\mathcal{X})$ and $K$ is  $\epsilon E_{d_{\cal X}}$-private  for (pseudo)metric $d$ on ${\cal X}$ then $K;K'$ is $\epsilon E_{d_{\cal X}}$-private. 
\end{restatable}

\section{Earth Mover's Privacy for bags of vectors in $\mathbb{R}^n$}\label{s1542}

\begin{figure}[htbp]
\centering
\begin{minipage}{.46\textwidth}
  \centering
  \includegraphics[width=.8\linewidth]{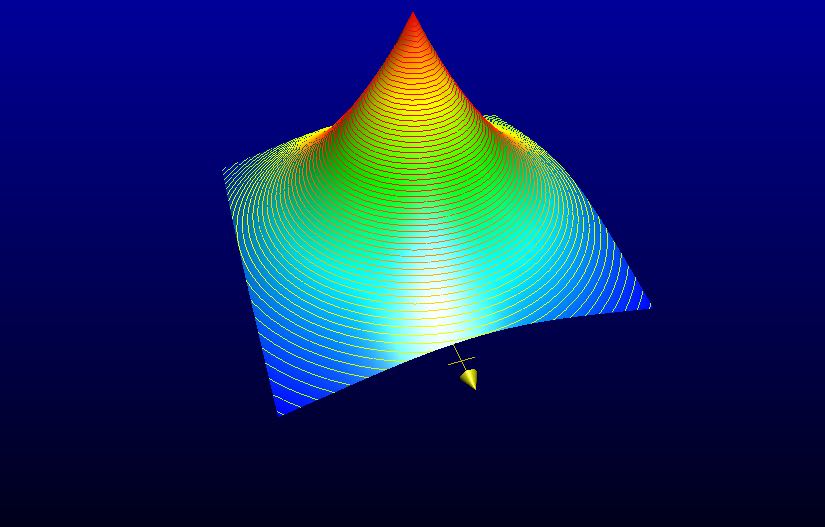}
  
  {3D plot}
\end{minipage}%
\begin{minipage}{.46\textwidth}
  \centering
  \includegraphics[width=.8\linewidth]{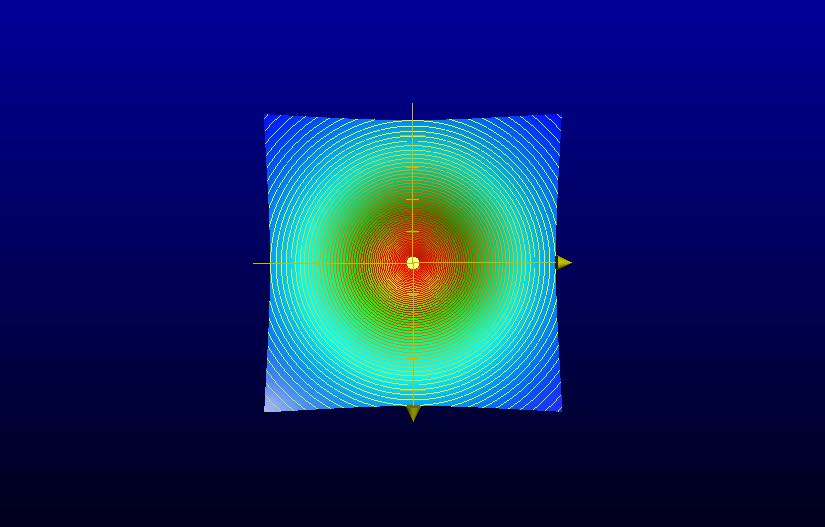}
  {\hspace{0.1cm}Contour diagram}
\end{minipage}
%
\caption{Laplace density function $\Lap{2}{\epsilon}$ in $\Real^2$}\label{f1059}
\end{figure}

In \Thm{emprivacy} we have shown how to promote a privacy mechanism on components to $E_{d_{\cal X}}$-privacy on a bag of those components. In this section we show how to implement a privacy mechanism satisfying \Eqn{e0848}, when the components are represented by high dimensional vectors in $\Real^n$ and the underlying metric is taken Euclidean on $\Real^n$, which we denote by $\Euclidean{\cdot}$.  

\medskip

We begin by summarising the basic probabilistic tools we need.  A \emph{probability density function} (PDF) over some domain ${\cal D}$ is a function $\phi: {\cal D} \Fun [0,1]$ whose value $\phi(z)$ gives the ``relative likelihood'' of $z$. 
The probability density function is used to compute the probability of an outcome ``$z\in A$'', for some region $A \subseteq {\cal D}$ as follows:
\begin{equation}\label{e1736}
\int_A \phi(x) ~dx~. 
\end{equation}

 In differential privacy, a popular density function used for implementing mechanisms is the \emph{Laplacian}, defined next.

\begin{definition}\label{d1141}
Let $n\geq 0$ be an integer $\epsilon>0$ be a real, and $v {\in} \Real^n$. We define the Laplacian probability density function in $n$-dimensions: 
\[
\Lap{n}{\epsilon}({v}) \Wide{\Defs} c_n^\epsilon{\times}e^{-\epsilon |\!|v|\!|}~,
\]
where $|\!|v|\!| = \sqrt{(v_1^2 + \dots +v_n^2)}$, and  $c_n^\epsilon$ is a real-valued constant satisfying the integral equation $1= \int \dots \int_{\Real^n} \Lap{n}{\epsilon}({v}) dv_1 \dots dv_n$.
\end{definition}

When $n=1$, we can compute $c_1^\epsilon= \epsilon/2$, and when $n=2$, we have that $c_2^\epsilon=\epsilon^2/2\pi$.

In privacy mechanisms, probability density functions are used to produce a ``noisy'' version of the released data. The benefit of the Laplace distribution is that, besides creating randomness, the likelihood that the released value is different from the true value decreases exponentially. This implies that the utility of the data release is high, whilst at the same time masking its actual value.  In \Fig{f1059} the probability density function $\Lap{2}{\epsilon}(v)$ depicts this situation, where we see that the highest relative likelihood of a randomly selected point on the plane being close to the origin, with the chance of choosing more distant points diminishing rapidly.  Once we are able to select a vector $v'$ in $\Real^n$ according to $\Lap{n}{\epsilon}$, we can ``add noise'' to any given vector $v$ as $v{+}v'$, so that the true value $v$ is highly likely to be perturbed only a small amount.   

\medskip


In order to use the Laplacian in \Def{d1141}, we need to implement it.  Andr\'es et al.\ \citep{andres2013geo} exhibited a mechanism  for $\Lap{2}{\epsilon}(v)$, and here we show how to extend that idea to the general case. The main idea of the construction for  $\Lap{2}{\epsilon}(v)$ uses the fact that any vector on the plane can be represented by spherical coordinates $(r, \theta)$, so that the probability of selecting a vector distance no more than $r$ from the origin can be achieved by selecting $r$ and $\theta$ independently. In order to obtain a distribution which overall is equivalent to $\Lap{2}{\epsilon}(v)$, Andr\'es et al.\ computed that $r$ must be selected according to a well-known distribution called the ``Lambert W'' function, and $\theta$ is selected uniformly over the unit circle. In our generalisation to  $\Lap{n}{\epsilon}(v)$, we observe that the same idea is valid \citep{boisbunon2012class}. Observe first that every vector in $\Real^n$ can be expressed as a pair $(r, p)$, where $r$ is the distance from the origin, and $p$ is a point in $B^n$, the unit \emph{hypersphere} in $\Real^n$. Now selecting vectors according to $\Lap{n}{\epsilon}(v)$ can be achieved by independently selecting $r$ and $p$, but this time 
$r$ must be selected according to the \emph{Gamma distribution}, and $p$ must be selected uniformly over $B^n$. We set out the details next.

\begin{definition}\label{d1732}
The \emph{Gamma distribution} of  (integer) shape $n$ and scale $\delta>0$ is determined by the probability density function:
\begin{equation}\label{e1745}
\Gam{n}{\delta}(r) \Wide{\Defs} \frac{ r^{n{-}1}e^{-r/\delta}}{\delta^n(n{-}1!)}~.
\end{equation}
\end{definition}

\begin{definition}\label{d1805}
The uniform distribution over the surface of the unit hypersphere $B^n$ is determined by the probability density function:
\begin{equation}\label{e1746}
\Unif{n}(v) \Wide{\Defs}  ~\frac{\Gamma(\frac{n}{2})}{n\pi^{n/2}} ~\textit{if} ~~ v \in B^n ~~ \textit{else} ~~ 0~,
\end{equation}
where $B^n \Defs \{v\in \Real^n ~|~  |\!|v|\!| =1\}$, and $\Gamma(\alpha) \Defs \int_0^\infty x^{\alpha{-}1}e^{-x} ~dx$ is the ``Gamma function''.
\end{definition}

With \Def{d1732} and \Def{d1805} we are able to provide an implementation of a mechanism which produces noisy vectors around a given vector in $\Real^n$ according to the Laplacian distribution in \Def{d1141}.  The first task is to show that our decomposition of  $\Lap{n}{\epsilon}$ is correct.


\begin{restatable}{lemma}{laplemma}
\label{l1328}
The $n$-dimensional Laplacian $\Lap{n}{\epsilon}(v)$ can be realised by selecting vectors represented as $(r, p)$, where $r$  is selected according to  $\Gam{n}{1/\epsilon}(r)$ and $p$ is selected independently according to $\Unif{n}(p)$.
\end{restatable}
\begin{proof} (Sketch)
The proof follows by changing variables to spherical coordinates and then showing that $\int_A \Lap{n}{\epsilon}(v)~dv$ can be expressed as the product of independent selections of $r$ and $p$.

We use a spherical-coordinate representation of $v$ as:
\[
\begin{array}{l}
r \Gets |\!|v|\!|~, ~\textit{and}~ \\
v_1 \Gets r \cos\theta_1~,~  v_2 \Gets r \sin\theta_1 \cos\theta_2 ~, \dots     v_n \Gets r \sin\theta_1 \sin\theta_2 \ldots ~,\sin\theta_{n-2} \sin\theta_{n{-}1} ~.
\end{array}
\]

Next we assume for simplicity that $A$ is a hypersphere of radius $R$; with that we can reason:
\begin{Reason}
\Step{}
{\int_A \Lap{n}{\epsilon}(v)~dv}
\WideStepR{$=$}{\Def{d1141}; $A$ is a hypersphere}
{\int_{|\!|v|\!| \leq R} c_n^\epsilon{\times}e^{-\epsilon |{v}|} ~dv}
\WideStepR{$=$}{$|\!|v|\!| = \sqrt{v_1^2 + \dots + v_n^2}$}
{\int_{|\!|v|\!| \leq R} c_n^\epsilon{\times}e^{-\epsilon \sqrt{v_1^2 + \dots + v_n^2}} ~dv}
\WideStepR{$=$}{Change of variables to spherical coordinates; see below \Eqn{e1426}}
{\int_{r \leq R} \int _{A_{\theta}}c_n^\epsilon{\times}e^{-\epsilon r}\frac{\partial(z_1, z_2, \ldots, z_n)}{\partial(r, \theta_1, \ldots, \theta_{n-1})}  ~dr d\theta_1\dots d\theta_{n{-}1}}
\WideStepR{$=$}{See below \Eqn{e1426}}
{\int_{r \leq R} \int _{A_{\theta}}c_n^\epsilon{\times}e^{-\epsilon r}  r^{n-1} \sin^{n-2}\theta_1 \sin^{n-3}\theta_2 \ldots \sin^2\theta_{n-3} \sin\theta_{n-2} ~dr d\theta_1\dots d\theta_{n{-}1}
~.}
\end{Reason}

Now rearranging we can see that this becomes a product of two integrals. The first $\int_{r\leq R} e^{-\epsilon r}  r^{n-1}$ is over the radius, and is proportional to the integral of the Gamma  distribution \Def{d1732}; and the second is an integral over the angular coordinates and is proportional to the surface of the unit hypersphere, and corresponds to the PDF at \Eqn{d1805}. We complete the details in the appendix \App{A0476}.
Finally, for the ``see below's'' we are using the ``Jacobian'' with details given at \App{A0476}: 
\begin{equation}\label{e1426}
    \frac{\partial(z_1, z_2, \ldots, z_n)}{\partial(r, \theta_1, \ldots, \theta_{n-1})} = r^{n-1} \sin^{n-2}\theta_1 \sin^{n-3}\theta_2 \ldots 
\end{equation}
\end{proof}

\medskip

We can now assemble the facts to demonstrate the n-Dimensional Laplacian.


\begin{restatable}[n-Dimensional Laplacian]{theorem}{laplace}
\label{thm:laplace}
Given $\epsilon >0$ and $n \in \mathbb{Z}^+$, let $K: \mathbb{R}^n \rightarrow \mathbb{D}\mathbb{R}^n$ be a mechanism that, given a vector $x \in \Real^n$ outputs a noisy value as follows:
 \[
 x \stackrel{K}{\longmapsto} x + x'
 \]
where $x'$ is represented as $(r, p)$ with $r\geq 0$, distributed according to $\Gam{n}{1/\epsilon}(r)$ and $p \in B^n$ distributed according to $\Unif{n}(p)$. 
Then  $K$ satisfies \Eqn{e0848} from \Thm{emprivacy}, i.e.\  $K$ satisfies $\epsilon |\!|\!\cdot \!|\!|$-privacy where $|\!|\!\cdot \!|\!|$ is the Euclidean metric on $\Real^n$.
\end{restatable}

\begin{proof}
(Sketch) Let $z, y \in \Real^n$. We need to show that for any (measurable) set $A\subseteq \Real^n$ that:
\begin{equation}\label{e1447}
K(z)(A)/ K(y)(A) \Wide{\leq} e^{\epsilon  |\!| z{-}y \!|\!|}~.
\end{equation}
However \Eqn{e1447} follows provided that the probability densities of respectively $K(z)$ and $K(y)$ satisfy it. By \Lem{l1328} the  probability density of $K(z)$, as a function of $x$ is distributed as $\Lap{n}{\epsilon}(z{-}x)$; and similarly for the probability density of $K(y)$. Hence we reason:
\begin{Reason}
\Step{}
{\Lap{n}{\epsilon}(z{-}x)/ \Lap{n}{\epsilon}(y{-}x)}
\StepR{$=$}{\Def{d1141}}
{c_n^\epsilon{\times}e^{-\epsilon  |\!| z{-}x \!|\!|}/{c_n^\epsilon{\times}e^{-\epsilon  |\!| y{-}x \!|\!|}}}
\StepR{$=$}{Arithmetic}
{e^{-\epsilon  |\!| z{-}x \!|\!|}\times e^{\epsilon  |\!| y{-}x \!|\!|}}
\StepR{$\leq$}{Triangle inequality; $s \mapsto e^s$ is monotone}
{e^{\epsilon  |\!| z{-}y \!|\!|}~,}
\end{Reason}
as required.
\end{proof}

\Thm{thm:laplace} reduces the problem of adding Laplace noise to vectors in $\Real^n$ to selecting a real value according to the Gamma distribution and an independent uniform selection of a unit vector.
Several methods have been proposed for generating random variables according to the Gamma distribution~\citep{kroese2013handbook} as well as for the uniform selection of vectors on the unit n-sphere~\citep{marsaglia1972choosing}.  The uniform selection of a unit vector  has also been described in \citep{marsaglia1972choosing}; it avoids the transformation to spherical coordinates by selecting $n$ random variables from the standard normal distribution to produce vector $v \in \Real^n$, and then normalising to output $\frac{v}{|v|}$. 


\subsection{Earth Mover's Privacy in ${\mathbb B}\mathbb{R}^n$}

Using the $n$-dimensional Laplacian, we can now implement an algorithm for $\epsilon N E_{\Euclidean{\cdot}}$-privacy. Algorithm~\ref{a1354} takes a bag of $n$-dimensional vectors as input and applies the $n$-dimensional Laplacian mechanism described in \Thm{thm:laplace} to each vector in the bag, producing a noisy bag of $n$-dimensional vectors as output. \Cor{c1135} summarises the privacy guarantee.


\begin{algorithm}
\caption{Earth Mover's Privacy Mechanism \label{a1354}}
\begin{algorithmic}[1]
\Require  vector $v$, dimension $n$, epsilon $\epsilon$ 
\Procedure{GenerateNoisyVector}{$v, n, \epsilon$} 
\State $r \gets Gamma(n, \frac{1}{\epsilon}$)
\State $u \gets \mathcal{U}(n)$
\State \textbf{return} $v + ru$
\EndProcedure
\end{algorithmic}

\begin{algorithmic}[1]
\Require bag $X$, dimension $n$, epsilon $\epsilon$
\Procedure{GeneratePrivateBag}{$X, n, \epsilon$}
\State $Z \gets ()$
\ForAll{ $x \in X$}
\State  $z \gets \Call{GenerateNoisyVector}{x, n, \epsilon}$ 
\State \text{add} $z$ \text{to} $Z$
\EndFor
\State \text{return} $Z$
\EndProcedure
\end{algorithmic}
\end{algorithm}

\begin{corollary}\label{c1135}
Algorithm~\ref{a1354} satisfies $\epsilon N E_{\Euclidean{\cdot}}$-privacy, relative to any two bags in ${\mathbb B}\Real^n$ of size $N$.
\begin{proof}
Follows from \Thm{emprivacy} and \Thm{thm:laplace}. 
\end{proof}

\end{corollary}

\subsection{Utility Bounds}

We prove a lower bound on the utility for this algorithm, which applies for high dimensional data representations. Given an output element $x$, we define $Z$ to be the set of outputs within distance $\Delta > 0$ from $x$. Recall that the distance function is a measure of utility, therefore $Z = \{ z ~|~ E_{|\!|\cdot |\!|}(x, z) \le \Delta\}$ represents the set of vectors within utility $\Delta$ of $x$. Then we have the following:

\begin{theorem}\label{l1534}
Given an input bag $b$ consisting of $N$ $n$-dimensional vectors, the mechanism defined by Algorithm \ref{a1354} outputs an element from $Z = \{ z ~|~ E_{|\!|\cdot |\!|}(b, z) \le \Delta\}$ with probability at least 
\[
1- e^{-\epsilon N \Delta}e_{n{-}1}(\epsilon N\Delta)~,
\]
whenever $\epsilon N\Delta \leq n/e$. (Recall that $e_k(\alpha) = \sum_{0\leq i \leq k} \frac{\alpha^i}{i!}$, the sum of the first $k+1$ terms in the series for $e^\alpha$~.)

%
\begin{proof} (Sketch)
Let $\underline{b} \in (\Real^n)^N$ be a (fixed) vector representation of the bag $b$. For $v\in (\Real^n)^N$, let  $v^\circ \in {\mathbb B}\Real^n$ be the bag comprising the $N$ components if $v$.  Observe that $NE_{|\!|\cdot |\!|}(b, v^\circ) \leq  M_{|\!|\cdot |\!|}(\underline{b}, v)$, and so  
\begin{equation}\label{e1208}
Z_M = \{ v~ |~ M_{|\!|\cdot |\!|}(\underline{b}, v) \le N\Delta\} \Wide{\subseteq} \{ v ~|~ E_{|\!|\cdot |\!|}(b, v^\circ) \le \Delta\} = Z_E~.
\end{equation}
 Thus the probability of outputting an element of $Z$ is the same as the probability of outputting $Z_E$, and by \Eqn{e1208} that is at least the probability of outputting an element from $Z_M$ by applying a standard n-dimensional Laplace mechanism  to each of the components of $\underline{b}$.  We can now compute:
\begin{Reason}
\Step{}
{\textit{Probability of outputting an element in} ~ Z_E}
\WideStepR{$\geq$}{\Eqn{e1208}}
{\int \dots \int_{v \in Z_M}\prod_{1\leq i\leq N} \Lap{n}{\epsilon}(\underline{b}_i {-} v_i) dv_1\dots dv_N} 
\WideStepR{$=$}{\Lem{l1328}}
{\int\dots\int_{v \in Z_M}\prod_{1\leq i\leq N} c_n^\epsilon  e^{-\epsilon |\!| \underline{b}_i {-} v_i|\!|} dv_1\dots dv_N~.} 
\end{Reason}
The result follows by completing the multiple integrals and applying some approximations, whilst observing that the variables in the integration are $n$-dimensional vector valued. The details appear in \App{A0476}.
\end{proof}

We note that our application word embeddings are typically mapped to vectors in $\Real^{300}$, thus we would use $n\sim 300$ in \Thm{l1534}.

\end{theorem}


\section{Text Document Privacy}\label{s1210}

In this section we bring everything together, and present a privacy mechanism for text documents; we explore how it contributes to the author obfuscation task described above.  Algorithm~\ref{a1453} describes the complete procedure for taking a document as a  bag-of-words, and outputting a ``noisy'' bag-of-words. Depending on the setting of parameter $\epsilon$, the output bag will be likely to be classified to be on a  similar topic as the input. 

\begin{algorithm}
\caption{Document privacy mechanism \label{a1453}}
\begin{algorithmic}[1]
\Require  Bag-of-words $b$, dimension $n$, epsilon $\epsilon$, Word embedding $\myVec: {\cal S} \Fun \Real^n$ 
\Procedure{GenerateNoisyBagOfWords}{$b, n, \epsilon, \myVec$} 
\State  $X \gets  \myVec^\star(b) $
\State  $Z \gets   \Call{GeneratePrivateBag}{X, n, \epsilon}$
\State \textbf{return} $(\myVec^{-1})^\star(Z)$
\EndProcedure
\end{algorithmic}

\medskip

Note that $\myVec^\star : {\mathbb B}{\cal S} \Fun  {\mathbb B}\Real^n$ applies $\myVec$ to each word in a bag $b$,
and $(\myVec^{-1})^\star  : {\mathbb B}\Real^n \Fun  {\mathbb B}{\cal S}$ reverses this procedure as a post-processing step; this involves determining the word $w$ that minimises the Euclidean distance $|\!|z- \myVec(w)|\!|$ for each $z$ in $Z$.
\end{algorithm}

Algorithm~\ref{a1453} uses a function $\myVec$ to turn the input document into a bag of word embeddings; next Algorithm~\ref{a1354} produces a noisy bag of word embeddings, and, in a final step the inverse $\myVec^{-1}$ is used to reconstruct an actual bag-of-words as output. In our implementation of Algorithm~\ref{a1453}, described below, we compute $\myVec^{-1}(x)$ to be the word $w$ that minimises the Euclidean distance $|\!|z- \myVec(w)|\!|$. The next result summarises the privacy guarantee for Algorithm~\ref{a1453}.

\begin{theorem}\label{t1355}
Algorithm \ref{a1453} satisfies $\epsilon NE_{d_{\cal S}}$-privacy, where $d_{\cal S}= \myDist_{\myVec}$.  That is to say: given input documents (bags) $b, b'$ both of size $N$, and $c$ a possible output bag, define the following quantities as follows:
$
k \Defs E_{|\!|\cdot |\!|}(\myVec^\star(b), \myVec^\star(b'))~, 
$
 $pr(b, c)$ and $pr(b', c)$ are the respective probabilities that $c$  is output given the input was $b$ or $b'$. Then:
\[
pr(b, c) \Wide{\leq} e^{\epsilon N k} \times pr(b', c)~.
\]
\begin{proof}
The result follows by appeal to \Thm{thm:laplace} for privacy on the word embeddings; the step to apply $\myVec^{-1}$ to each vector is a post-processing step which by \Lem{l1112} preserves the privacy guarantee. 
\end{proof}
\end{theorem}

Although \Thm{t1355} utilises ideas from differential privacy, an interesting question to ask is how it contributes to the PAN@Clef author obfuscation task, which recall asked for mechanisms that preserve content but mask features that distinguish authorship.  Algorithm~\ref{a1453} does indeed attempt to preserve content (to the extent that the topic can still be determined) but it does not directly ``remove stylistic features''. So has it, in fact, disguised the author's characteristic style?  To answer that question, we review \Thm{t1355} and interpret what it tells us in relation to author obfuscation. The theorem implies that it is indeed possible to make the (probabilistic) output from two distinct documents $b, b'$ almost indistinguishable by choosing $\epsilon$ to be extremely small in comparison with $N{\times}E_{|\!|\cdot |\!|}(\myVec^{\star}(b), \myVec^{\star}(b'))$. However, if $E_{|\!|\cdot |\!|}(\myVec^{\star}(b), \myVec^{\star}(b'))$ is very large -- meaning that $b$ and $b'$ are on entirely different topics, then $\epsilon$ would need to be so tiny that the noisy output document would be highly unlikely to be on a topic remotely close to either $b$ or $b'$ (recall \Lem{l1534}).  

This observation is actually highlighting the fact that, in some circumstances, the topic itself is actually a feature that characterises author identity.  (First-hand accounts of breaking the world record
for highest and longest free fall jump would immediately narrow the field down to the title holder.)  This means that \emph{any} obfuscating mechanism would, as for Algorithm~\ref{a1453}, only be able to obfuscate documents so as to disguise the author's identity if there are several authors who write on similar topics.  And it is in that spirit, that we have made the first step towards a satisfactory obfuscating mechanism:  provided that documents are similar in topic (i.e.\ are close when their embeddings are measured by $E_{|\!|\cdot |\!|}$) they can be obfuscated so that it is unlikely that the content is disturbed, but that the contributing authors cannot be determined easily.

We can see the importance of the ``indistinguishability'' property wrt.\ the PAN obfuscation task. In stylometry analysis the representation of words for eg.\ author classification is completely different to the word embeddings which have used for topic classification. State-of-the-art author attribution algorithms represent words as ``character n-grams'' \cite{koppel2011authorship} which have been found to capture stylistic clues such as systematic spelling errors. A \emph{character 3-gram} for example represents  a given word as the complete list of substrings of length 3. For example character 3-gram representations of ``color'' and ``colour'' are:

\begin{itemize}
    \item[$\cdot$] ``color'' $\mapsto$ $|\![$ ``col'', ``olo'', ``lor''  $]\!|$
        \item[$\cdot$] ``colour'' $\mapsto$ $|\![$ ``col'', ``olo'', ``lou'', ``our''  $]\!|$
\end{itemize}

For author identification, any output from Algorithm~\ref{a1453} would then need to be further transformed to a bag of character n-grams, as a post processing step; by \Lem{l1112} this additional transformation preserves the privacy properties of Algorithm~\ref{a1453}.
We explore this experimentally in the next section.

\section{Experimental Results}\label{s1545}





\paragraph{Document Set}
The PAN@Clef tasks and other similar work have used a variety of types of text for author identification and author obfuscation.  Our desiderata are that we have multiple authors writing on one topic (so as to minimise the ability of an author identification system to use topic-related cues) and to have more than one topic (so that we can evaluate utility in terms of accuracy of topic classification).  Further, we would like to use data from a domain where there are potentially large quantities of text available, and where it is already annotated with author and topic.

Given these considerations, we chose ``fan fiction'' as our domain.  Wikipedia defines \emph{fan fiction} as follows: ``Fan fiction \ldots\ is fiction about characters or settings from an original work of fiction, created by fans of that work rather than by its creator.''  This is also the domain that was used in the PAN@Clef 2018 author attribution challenge,\footnote{\url{https://pan.webis.de/clef18/pan18-web/author-identification.html}} although for this work we scraped our own dataset.  We chose one of the largest fan fiction sites and the two largest ``fandoms'' there;\footnote{\url{https://www.fanfiction.net/book/}, with the two largest fandoms being Harry Potter (797,000 stories) and Twilight (220,000 stories).} these fandoms are our topics.  We scraped the stories from these fandoms, the largest proportion of which are for use in training our topic classification model.  We held out two subsets of size 20 and 50, evenly split between fandoms/topics, for the evaluation of our privacy mechanism.\footnote{Our Algorithm~\ref{a1453} is computationally quite expensive, because each word $w = \myVec^{-1}(x)$ requires the calculation of Euclidean distance with respect to the whole vocabulary.  We thus use relatively small evaluation sets, as we apply the algorithm to them for multiple values of $\epsilon$.}  We follow the evaluation framework of \cite{koppel2011authorship}: for each author we construct an known-author \textsc{text} and an unknown-author \textsc{snippet} that we have to match to an author on the basis of the known-author texts.  (See Appendix~\ref{app:data} for more detail.)

\paragraph{Word Embeddings}
There are sets of word embeddings trained on large datasets that have been made publicly available.  Most of these, however, are already normalised, which makes them unsuitable for our method.  We therefore use the Google News word2vec embeddings as the only large-scale unnormalised embeddings available.  (See Appendix~\ref{app:data} for more detail.)

\paragraph{Inference Mechanisms}

We have two sorts of machine learning inference mechanisms: our adversary mechanism for author identification, and our utility-related mechanism for topic classification.  For each of these, we can define inference mechanisms both within the same representational space or in a different representational space.  As we noted above, in practice both author identification adversary and topic classification will use different representations, but examining same-representation inference mechanisms can give an insight into what is happening within that space.


\paragraph{Different-representation author identification}
For this we use the algorithm by \cite{koppel2011authorship}.
This algorithm is widely used: it underpins two of the winners of PAN shared tasks \citep{seidman:2013:PAN,khonji-iraqi:2014:PAN}; is a common benchmark or starting point for other methods
\citep{sapkota-etal:2015:NAACL,ruder-etal:2016:arxiv,halvani-etal:2017:arxiv,potha-stamatatos:2017:CLEF}; and is a standard inference attacker for the PAN shared task on
authorship obfuscation.\footnote{\url{http://pan.webis.de/clef17/pan17-web/author-obfuscation.html}}  It works by representing each text as a vector of space-separated character n-gram counts, and comparing repeatedly sampled subvectors of known-author texts and snippets using cosine similarity.  We use as a starting point the code from a reproducibility study \cite{potthast-etal:2016:ECIR}, but have modified it to improve efficiency.  (See Appendix~\ref{app:impl} for more details.)

\paragraph{Different-representation topic classification}
Here we choose fastText \cite{bojanowski-etal:2016:arxiv,joulin-etal:2016:arxiv}, a high-performing supervised machine learning classification system.  It also works with word embeddings; these differ from word2vec in that they are derived from embeddings over character n-grams, learnt using the same skipgram model as word2vec.  This means it is able to compute representations for words that do not appear in the training data, which is helpful when training with relatively small amounts of data; also useful when training with small amounts of data is the ability to start from pretrained embeddings trained on out-of-domain data that are then adapted to the in-domain (here, fan fiction) data.  After training, the accuracy on a validation set we construct from the data is 93.7\% (see Appendix~\ref{app:impl} for details).

\paragraph{Same-representation author identification}
In the space of our word2vec embeddings, we can define an inference mechanism that for an unknown-author snippet chooses the closest known-author text by Euclidean distance.

\paragraph{Same-representation topic classification}
Similarly, we can define an inference mechanism that considers the topic classes of neighbours and predicts a class for the snippet based on that.  This is essentially the standard  $k$ ``Nearest Neighbours'' technique ($k$-NN) \cite{hastie-etal:2009}, a non-parametric method that assigns the majority class of the $k$ nearest neighbours.  1-NN corresponds to classification based on a Voronoi tesselation of the space, has low bias and high variance, and asymptotically has an error rate that is never more than twice the Bayes rate; higher values of $k$ have a smoothing effect.  Because of the nature of word embeddings, we would not expect this classification to be as accurate as the fastText classification above: in high-dimensional Euclidean space (as here), almost all points are approximately equidistant.  Nevertheless, it can give an idea about how a snippet with varying levels of noise added is being shifted in Euclidean space with respect to other texts in the same topic.  Here, we use $k=5$.  Same-representation author identification can then be viewed as 1-NN with author as class.

{\tiny 
\begin{table}
    \centering
    \begin{tabular}{|l||r|r||r|r|}
    \cline{2-5}
    \multicolumn{1}{c}{} & \multicolumn{4}{|c|}{20-author set}\\
    \hline
$\epsilon$ & SRauth & SRtopic & DRauth	& DRtopic\\
\hline
\hline
none &	12 &	16 &	15 &	18\\
\hline
30 &	8 &	18 &	16 &	18\\
25 &	8 &	18 &	14 &	17\\
20 &	5 &	11 &	11 &	16\\
15 &	2 &	11 &	12 &	17\\
10 &	0 &	15 &	11 &	19\\
\hline
    \end{tabular}
    \hspace{0.25cm}
    \begin{tabular}{|l||r|r||r|r|}
    \cline{2-5}
    \multicolumn{1}{c}{} & \multicolumn{4}{|c|}{50-author set}\\
    \hline
$\epsilon$ & SRauth & SRtopic & DRauth	& DRtopic\\
\hline
\hline
none &	19 &	36 &	27 &	43\\
\hline
30 &	19 &	37 &	29 &	43\\
25 &    17 &    34 &    24 &    41\\				
20 &	12 &	28 &	19 &	42\\
15 &    9 &     22 &    13 & 42\\				
10 &	1 &	24 &	10 &	43\\
\hline
    \end{tabular}
    
    \vspace{0.2cm}
    
  \caption{Number of correct predictions of author/topic in the 20-author set (left) and 50-author set (right), using 1-NN for same-representation author identification (SRauth), 5-NN for same-representation topic classification (SRtopic), the Koppel algorithm for different-representation author identification (DRauth) and fastText for different-representation topic classification (DRtopic). \label{tab:author-results}}
\end{table}
}

\paragraph{Results:} Table~\ref{tab:author-results} contains the results for both document sets, for the unmodified snippets (``none'') or with the privacy mechanism of Algorithm~\ref{a1453} applied with various levels of $\epsilon$: we give results for $\epsilon$ between 10 and 30, as at $\epsilon=40$ the text does not change, while at $\epsilon=1$ the text is unrecognisable.  For the 20-author set, a random guess baseline would give 1 correct author prediction, and 10 correct topic predictions; for the 50-author set, these values are 1 and 25 respectively.

Performance on the unmodified snippets using different-representation inference mechanisms is quite good: author identification gets 15/20 correct for the 20-author set and 27/50 for the 50-author set; and topic classification 18/20 and 43/50 (comparable to the validation set accuracy, although slightly lower, which is to be expected given that the texts are much shorter).  For various levels of $\epsilon$, with our different-representation inference mechanisms we see broadly the behaviour we expected: the performance of author identification drops, while topic classification holds roughly constant.  Author identification here does not drop to chance levels: we speculate that this is because (in spite of our choice of dataset for this purpose) there are still some topic clues that the algorithm of \cite{koppel2011authorship} takes advantage of: one author of Harry Potter fan fiction might prefer to write about a particular character (e.g. Severus Snape), and as these character names are not in our word2vec vocabulary, they are not replaced by the privacy mechanism.

In our same-representation author identification, though, we do find performance starting relatively high (although not as high as the different-representation algorithm) and then dropping to (worse than) chance, which is the level we would expect for our privacy mechanism.  The $k$-NN topic classification, however, shows some instability, which is probably an artefact of the problems it faces with high-dimensional Euclidean spaces.  (We show a sample of texts and nearest neighbours in Appendix~\ref{app:analysis})

\section{Related Work}


\paragraph{Author Obfuscation} The most similar work to ours is by Weggenmann and Kerschbaum \citep{weggenmann2018syntf} who also consider the author obfuscation problem but apply standard differential privacy using a Hamming distance of 1 between all documents. As with our approach, they consider the simplified utility requirement of topic preservation and use word embeddings to represent documents. Our approach differs in our use of the Earth Mover's metric to provide a strong utility measure for document similarity. 

An early work in this area by Kacmarcik et al. \cite{kacmarcik2006obfuscating} applies obfuscation by modifying the most important stylometric features of the text to reduce the effectiveness of author attribution. This approach was used in Anonymouth~\citep{mcdonald2012use}, a semi-automated tool that provides feedback to authors on which features to modify to effectively anonymise their texts. A similar approach was also followed by Karadhov et al.~\citep{karadzhov2017case} as part of the PAN@Clef 2017 task. 

Other approaches to author obfuscation, motivated by the PAN@Clef task, have focussed on the stronger utility requirement of semantic sensibility~\citep{bakhteev:2017,castro:2017,mansoorizadeh:2016b}. Privacy guarantees are therefore ad hoc and are designed to increase misclassification rates by the author attribution software used to test the mechanism. 

Most recently there has been interest in training neural networks models which can protect author identity whilst preserving the semantics of the original document~\citep{shetty2018a4nt, emmery2018style}. Other related deep learning methods aim to obscure other author attributes such as gender or age \cite{li-etal:2018:ACL,coavoux-etal:2018:EMNLP}.
While these methods produce strong empirical results, they provide no formal privacy guarantees.  Importantly, their goal also differs from the goal of our paper: they aim to obscure properties of authors in the \textit{training set} (with the intention of the author-obscured learned representations being made available), while we assume that an adversary may have access to raw training data to construct an inference mechanism with full knowledge of author properties, and in this context aim to hide the properties of some other text external to the training set.

\paragraph{Machine Learning and Differential Privacy}
Outside of author attribution, there is quite a body of work on introducing differential privacy to machine learning: \cite{dwork2014} gives an overview of a classical machine learning setting; more recent deep learning approaches include \cite{shokri-shmatikov:2015:CCS,abadi-etal:2016:CCS}.  However, these are generally applied in other domains such as image processing: text introduces additional complexity because of its discrete nature, in contrast to the continuous nature of neural networks.  A recent exception is \cite{mcmahan-etal:2018:ICLR}, which constructs a differentially private language model using a recurrent neural network; the goal here, as for instances above, is to hide properties of data items in the training set.

\paragraph{Generalised Differential Privacy} Also known as $d_\mathcal{X}$-privacy~\citep{chatzikokolakis2013broadening}, this definition was originally motivated by the problem of geo-location privacy~\citep{andres2013geo}. Despite its generality, $d_{\mathcal X}$-privacy has yet to find significant applications outside this domain; in particular, there have been no applications to text privacy.

\paragraph{Text Document Privacy} This typically refers to the sanitisation or redaction of documents either to protect the identity of individuals or to protect the confidentiality of their sensitive attributes. For example, a medical document may be modified to hide specifics in the medical history of a named patient. Similarly, a classified document may be redacted to protect the identity of an individual referred to in the text. 

Most approaches to sanitisation or redaction rely on first identifying sensitive terms in the text, and then modifying (or deleting) only these terms to produce a sanitised document. Abril et al.~\citep{abril2011declassification} proposed this two-step approach, focussing on identification of terms using NLP techniques.  Cumby and Ghani~\citep{cumby2011machine} proposed $k-confusability$, inspired by $k-anonymity$~\citep{sweeney2002}, to perturb sensitive terms in a document so that its (utility) class is confusable with at least $k$ other classes. Their approach requires a complete dataset of similar documents for computing (mis)classification probabilities. Anandan et al.~\citep{anandan2012t} proposed \emph{t-plausibility} which generalises sensitive terms such that any document could have been generated from at least $t$ other documents. S{\'a}nchez and Batet~\citep{sanchez2016c} proposed \emph{C-sanitisation}, a model for both detection and protection of sensitive terms ($C$) using information theoretic guarantees. In particular, a \emph{C-sanitised} document should contain no collection of terms which can be used to infer any of the sensitive terms. 

Finally, there has been some work on noise-addition techniques in this area. Rodriguez-Garcia et al.~\citep{rodriguez2015semantic} propose semantic noise, which perturbs sensitive terms in a document using a distance measure over the directed graph representing a predefined ontology.

Whilst these approaches have strong utility, our primary point of difference is our insistence on a differential privacy-based guarantee. This ensures that every output document could have been produced from any input document with some probability, giving the strongest possible notion of plausible-deniability.

\section{Conclusions}\label{s0939}
We have shown how to combine representations of text documents with generalised differential privacy in order to implement a privacy mechanism for text documents. Unlike most other techniques for 
privacy in text processing, ours provides a guarantee in the style of differential privacy. Moreover we have demonstrated experimentally the trade off between utility and privacy.

This represents an important step towards the implementation of privacy mechanisms that could produce readable summaries of documents with a privacy guarantee. One way to achieve this goal would be to reconstruct readable documents from the bag-of-words output that our mechanism currently provides. A range of promising techniques for reconstructing readable texts from bag-of-words have already produced some good experimental results \cite{wan-etal:2009:EACL,zhang-clark:2015:CL,hasler-etal:2017:INLG}. In future work we aim to explore how techniques such as these could  be applied as a final post processing step for our mechanism.

\bibliography{MathBib,NLPexper}


\newpage

\appendix

\setcounter{theorem}{0}
\renewcommand{\thetheorem}{\Alph{section}\arabic{theorem}}

\setcounter{definition}{0}
\renewcommand{\thedefinition}{\Alph{section}\arabic{definition}}

\setcounter{lemma}{0}
\renewcommand{\thelemma}{\Alph{section}\arabic{lemma}}

\setcounter{corollary}{0}
\renewcommand{\thecorollary}{\Alph{section}\arabic{corollary}}

\section{Appendix A}

Here we present proofs omitted from the main body of the paper.

\subsection{Proofs Omitted from \Sec{s1520}}\label{A0768}

To prove \Thm{emprivacy} we introduce the following results.

\begin{definition}{}
An $n \times n$ matrix whose elements are non-negative and has all rows and columns summing to 1 is called \emph{doubly stochastic}. A doubly stochastic matrix which contains only 1's and 0's is called a \emph{permutation matrix}. 
\end{definition}

\begin{theorem}{(Birkhoff-von Neumann)}\label{t1244}
The set of $n \times n$ doubly stochastic matrices forms a convex polytope whose vertices are the $n \times n$ permutation matrices.
\end{theorem}

The Birkhoff-von Neumann theorem says that the set of doubly stochastic matrices is a closed, bounded convex set, and every doubly stochastic matrix can be written as a convex combination of the permutation matrices. We can now prove the following result.

\begin{lemma}
\label{thm:perm-matrix}
Let $D$ and $F$ be non-negative $n \times n$ matrices. Then the problem of finding an $F$ which minimises 
\begin{align*}
     \sum\limits_{i = 1}^{n} \sum\limits_{j = 1}^{n} D_{ij} F_{ij}
\end{align*}
subject to
\begin{align*}
       & \sum\limits_{i = 1}^{n} F_{ij} = 1  \quad \text{and } \quad \sum\limits_{j = 1}^{n} F_{ij} = 1   
\end{align*}
always has an $n \times n$ permutation matrix as an optimal solution.

\begin{proof}
We prove this by contradiction. Let $F^\star$ be an optimal $n \times n$ solution matrix. Since $F^\star$ is doubly stochastic we can apply Birkhoff-von Neumann. Firstly, we know that such a solution exists (since the set of solutions is closed and bounded). We now assume that $F^\star$ is not a permutation matrix, and also that no permutation matrix is optimal. Let $\{P^1, P^2, \ldots, P^k\}$ be the set of $n \times n$ permutation matrices. Then, by Birkhoff-von Neumann, we can write 
\begin{align}
    F^\star = \lambda_1 P^1 + \lambda_2 P^2 + \ldots + \lambda_k P^k
\label{eqn-bvn}
\end{align}
where $\lambda_i \ge 0$ and $\sum\limits_{i = 1}^k \lambda_i = 1$. Since $F^\star$ is optimal and none of the $P^i$ are optimal, by assumption we also know
\begin{align}\label{assump}
              \sum\limits_{i,j} P^m_{ij}D_{ij} &> \sum\limits_{i,j} F^\star_{ij}D_{ij}                                                   
\end{align}
for $0 < m \le k$. And thus we have:
\begin{Reason}
\Step{}
{\sum\limits_{i,j} F^\star_{ij} D_{ij}}
\StepR{$=$}{\Eqn{eqn-bvn}}
{\sum\limits_{i,j} (\lambda_1 P^1_{ij} + \ldots + \lambda_k P^k_{ij}) D_{ij}}
\StepR{$=$}{Factorising}
{\sum\limits_{i,j} \lambda_1 P^1_{ij} D_{ij} +  \ldots + \sum\limits_{i,j} \lambda_k P^k_{ij} D_{ij}}
\StepR{$>$}{\Eqn{assump}}
{\sum\limits_{i,j} \lambda_1 F^\star_{ij} D_{ij} + \ldots + \sum\limits_{i,j} \lambda_k F^\star_{ij} D_{ij}}
\StepR{$=$}{Arithmetic}
{\lambda_1  \sum\limits_{i,j} F^\star_{ij} D_{ij} + \ldots + \lambda_k \sum\limits_{i,j} F^\star_{ij} D_{ij}}
\StepR{$=$}{$\sum\limits_{i = 1}^k \lambda_i = 1$}                                      
{\sum\limits_{i,j} F^\star_{ij} D_{ij}} 
\end{Reason}
which is a contradiction. Thus, either $F^\star$ is a permutation matrix, or there must be a permutation matrix which is also optimal.
\end{proof}
\end{lemma}

Now, for permutation $\sigma \in \{1 \dots N\}\Fun \{1\dots N\}$ write $\underline{c}^\sigma$ to be the vector with components permuted by $\sigma$, so that $\underline{c}^\sigma_i = \underline{c}_{\sigma(i)}$.

\begin{lemma}{}
\label{lemma-extend}
Let $d_\mathcal{X}$ be a pseudometric on $\mathcal{X}$ and let $K: \mathcal{X} \rightarrow \mathbb{D}\mathcal{X}$ be a mechanism satisfying $d_\mathcal{X}$-privacy. Let $x, z \in \mathbb{B}\mathcal{X}$ be bags of length $N$ with corresponding vectors $\underline x, \underline z \in \mathcal{X}^N$. Then $K$ can be extended to a mechanism $K^\star: \mathbb{B}\mathcal{X} \rightarrow \mathbb{D}(\mathbb{B}\mathcal{X})$ satisfying:
\[
    K^\star(x)(\{z\}) = \sum_\sigma K(\underline x)(\{\underline z^{\sigma}\}) 
\]
where the sum is over unique permutations of elements in $\underline z$. 

\begin{proof}

Recall that a mechanism is a probabilistic function; we have to show that there is a mechanism $K^\star$ that outputs a valid distribution over bags in $\mathbb{B} \mathcal{X}$ given an input bag in $\mathcal{X}$. We show this by constructing the required mechanism.

We can easily extend $K$ to a mechanism $\underline K': \mathcal{X}^N \rightarrow \mathbb{D}(\mathcal{X}^N)$ operating on \emph{vectors} by applying $K$ to each element of $\underline x$ in order. That is,
\[
     \underline K'(\underline{x})(\{\underline{z}\}) = K(\underline x_1)(\underline z_1) \times K(\underline x_2)(\underline z_2) \times \ldots \times K(\underline x_n)(\underline z_n)
\]

$\underline K'(\underline x)$ defines a valid probability distribution for any $\underline x$ since we sum over all possible output vectors $\underline z$. 

Observe that the mechanism $\underline K'$ produces the same output distribution regardless of the ordering of elements in $\underline x$ (since the mechanism $K$ operates on each element independently). Therefore the distribution over bags depends only on the different permutations of elements in the output $\underline{z}$. That is, 
\[
    K^\star(x)(\{z\}) = \sum_\sigma K(\underline x_1)(\underline z_{\sigma(1)}) \times K(\underline x_2)(\underline z_{\sigma(2)}) \times \ldots \times K(\underline x_n)(\underline z_{\sigma(n)})
\]

Here $K^\star(x)$ also defines a valid probability distribution, since it produces the same distribution as $K'(\underline{x})$ except that the output probabilities are `collected' for all permutations of the output vector. Thus $K^\star$ is the required mechanism.
\qed
\end{proof}
\end{lemma}

We are now ready to prove \Thm{emprivacy}.\\
\\
\textbf{Theorem \ref{emprivacy}.}
\textit{
Let $d_{\mathcal{X}}$ be a pseudometric on $\mathcal{X}$ and let $K: \mathcal{X} \rightarrow \mathbb{D}\mathcal{X}$ be a mechanism satisfying $\epsilon d_{\cal X}$-privacy, i.e.
\begin{equation*}
K(x)(Z) \Wide{\leq} e^{\epsilon d_{\cal X}(x, x')}K(x')(Z)~,~ \textit{for all~} x, x'\in {\cal X}~Z \subseteq {\cal X}.
\end{equation*}
Let $K^*: \mathbb{B}\mathcal{X} \Fun \Dist( \mathbb{B}\mathcal{X})$ be the mechanism obtained by applying $K$ independently to each element of $X$ for any $X \in \mathbb{B}\mathcal{X}$. Denote by $K^\star\restrict$ the restriction of $K^\star$ to bags of fixed size $N$.  Then $K^\star\restrict$ satisfies $\epsilon N {E_{d_\mathcal{X}}}$-privacy.}

\begin{proof}
Let $b, b'$ be input bags of size $N$, and $c$ a possible output bag of $K^\star$. Observe that $c$ also has size $N$. Therefore, the Earth Mover's constraints in \Def{d0938} can be rewritten as:
\begin{align*}
& \sum\limits_{i = 1}^{N} F_{ij} = \frac{1}{|N|} \quad \textit{and}\quad
        \sum\limits_{j = 1}^{N} F_{ij} = \frac{1}{|N|}~, \quad  F_{ij} \ge 0, \quad1 \le i,j \le N  
\end{align*}
This has the same form as in \Lem{thm:perm-matrix}, thus the optimal transportation for $E_{d_\mathcal X}(b, c)$ is achieved by moving each word in bag $b$ to a single word in bag $c$. The same is true for $E_{d_\mathcal X}(b', c)$ and $E_{d_\mathcal X}(b, b')$.
Next, we fix $\underline b$ and $\underline b'$ to be vector representations of respectively $b, b'$ in $\mathcal{X}^N$ such that the optimal transportation for $E_{d_\mathcal X}(b, b')$ is
\begin{equation}\label{e1949}
\Earth{d_{\cal X}}(b, b') \Wide{=}1/N {\times}\sum_{1\leq i\leq N} d_{\cal X}(\underline{b}_i, \underline{b}'_i) \Wide{=} \Manhattan{d_{\cal X}}(\underline{b}, \underline{b}')/N~.
\end{equation}
That is, we fix the ordering of the elements in the vectors $\underline b, \underline b'$ so that the Manhattan distance is exactly the Earth Mover's distance (which we know can be done thanks to \Lem{thm:perm-matrix}). 
Finally, from \Lem{lemma-extend} we know that the following equality between probabilities holds:
\begin{equation}\label{e1239}
K^\star(b)\{c\} \Wide{=} \sum_{\sigma}{K}(\underline{b})\{\underline{c}^\sigma\}~,
\end{equation}
where the summation is over all permutations that give distinct vector representations of $c$. We now compute directly:

\begin{Reason}
\Step{}
{K^\star(b)\{c\}}
\StepR{$=$}{\Eqn{e1239} for $b, c$}
{\sum_{\sigma}{K}(\underline{b})\{\underline{c}^\sigma\}}
\StepR{$\leq$}{\Eqn{e1939} for $\underline{b}, \underline{b'}, \underline{c}$}
{\sum_{\sigma}e^{\epsilon\Manhattan{d}(\underline{b}, \underline{b}')}{K}(\underline{b}')\{\underline{c}^\sigma\}}
\StepR{$=$}{Arithmetic and \Eqn{e1949}}
{e^{\epsilon N\Earth{d}(\underline{b}, \underline{b}')}{\sum_{\sigma}{K}(\underline{b}')\{\underline{c}^\sigma\}}}
\StepR{$=$}{\Eqn{e1239} for $b', c$}
{e^{\epsilon N\Earth{d}(\underline{b}, \underline{b}')}K^\star(b')\{c\}~,}
\end{Reason}

Thus the mechanism $K^\star$ satisfies $\epsilon N \Earth{d_\mathcal{X}}$-privacy for singleton sets, and by extension for all finite sets $Z \subseteq \mathbb{B} \cal{X}$.

\qed
\end{proof}

\postlemma*

\begin{proof}
Let $b, c \in  \mathbb{B}\mathcal{X}$; we reason as follows.
\begin{Reason}
\Step{}
{(K ; K')(b)(Z)}
\StepR{$=$}{\Eqn{d1143}}
{\sum_{b' \In  \mathbb{B}\mathcal{X}} K(b)(\{ b'\}){\times} K'(b')(Z)}
\StepR{$\leq$}{$K$ is $\epsilon E_d$-private}
{\sum_{b' \In  \mathbb{B}\mathcal{X}} e^{\epsilon E_{d_{\cal X}}(b, c)}K(c)(\{ b'\}){\times} K'(b')(Z)}
\WideStepR{$=$}{\Eqn{d1143}; arithmetic}
{e^{\epsilon E_{d_{\cal X}}(b, c)}(K ; K')(c)(Z)~.}
\end{Reason}
\end{proof}


\subsection{Proofs Omitted from \Sec{s1542}}\label{A0476}


\laplemma*

\begin{proof}
We note first that the n-dimensional Laplacian is spherically symmetric; that is, we want the length of the random vector to follow a Laplacian distribution independently from its direction.  Therefore the Laplacian has a stochastic representation: 
\begin{align}
      X = RU
\end{align}
where $R = |\!|X|\!|$ and $U = X /  |\!|X|\!|$. i.e. $U$ is a random variable drawn from the uniform distribution on the $n$-sphere (that is, $U \sim \Unif{n}(p)$) and $R$ is a `scaling' component independent from $U$.

We now show that the radial component is drawn from the Gamma distribution. From \Def{d1141} we have that 
\[
  \int \dots \int_{\Real^n} \Lap{n}{\epsilon}({v}) dv_1 \dots dv_n = 1
\]
We perform a conversion to spherical co-ordinates $ (r, \theta_1, \theta_2, \ldots, \theta_{n-1})$ using the following transformation \citep{mustard1964}:
\begin{align*}
    v_1 &= r \cos\theta_1 \\
    v_2 &= r \sin\theta_1 \cos\theta_2 \\
    v_3 &= r \sin\theta_1 \sin\theta_2 \cos\theta_3 \\
    &\ldots \\
    v_{n-1} &= r \sin\theta_1 \sin\theta_2 \ldots \sin\theta_{n-2} \cos\theta_{n-1} \\
    v_n &= r \sin\theta_1 \sin\theta_2 \ldots \sin\theta_{n-2} \sin\theta_{n-1}     
\end{align*}
where $r = \sqrt{v_1^2 + \dots + v_n^2}$. The bounds for the new co-ordinates are:
\begin{align}\label{bounds}
    0 \le r < \infty ;~ 
    0 \le \theta_1 \dots \theta_{n-2} \le \pi ;~
    0 \le \theta_{n-1} \le 2 \pi
\end{align}
We also need the Jacobian determinant, denoted $\frac{\partial(v_1, v_2, \ldots, v_n)}{\partial(r, \theta_1, \ldots, \theta_{n-1})}$. This is well-known to be:
\begin{align}\label{eqnjac}
    \frac{\partial(v_1, v_2, \ldots, v_n)}{\partial(r, \theta_1, \ldots, \theta_{n-1})} = r^{n-1} \sin^{n-2}\theta_1 \sin^{n-3}\theta_2 \ldots \sin^2\theta_{n-3} \sin\theta_{n-2}
\end{align}
And therefore we reason:
\begin{Reason}
\Step{}
{ \int \dots \int_{\Real^n} \Lap{n}{\epsilon}({v}) dv_1 \dots dv_n}
\WideStepR{$=$}{\Def{d1141}; bounds}
{ \int_{-\infty}^{\infty} \dots \int_{-\infty}^{\infty} c_n^\epsilon{\times}e^{-\epsilon |\!|{v}|\!|} ~dv_1 \dots dv_n}
\WideStepR{$=$}{$|\!|v|\!| = \sqrt{v_1^2 + \dots + v_n^2}$}
{ \int_{-\infty}^{\infty} \dots \int_{-\infty}^{\infty}  c_n^\epsilon{\times}e^{-\epsilon \sqrt{v_1^2 + \dots + v_n^2}} ~dv_1 \dots dv_n}
\WideStepR{$=$}{Change of variables; \Eqn{bounds}}
{\int_{0}^{\infty} \int _{0}^{\pi} \dots \int_{0}^{2\pi} c_n^\epsilon{\times}e^{-\epsilon r}\frac{\partial(v_1, v_2, \ldots, v_n)}{\partial(r, \theta_1, \ldots, \theta_{n-1})}  ~dr d\theta_1\dots d\theta_{n{-}1}}
\WideStepR{$=$}{\Eqn{eqnjac}}
{\int_{0}^{\infty} \int _{0}^{\pi} \dots \int_{0}^{2\pi}c_n^\epsilon{\times}e^{-\epsilon r}  r^{n-1} \sin^{n-2}\theta_1 \ldots \sin^2\theta_{n-3} \sin\theta_{n-2} ~dr d\theta_1\dots d\theta_{n{-}1}}
\WideStepR{$=$}{Independent variables; $c_n^\epsilon = c_0 \times c_1 \times \dots \times c_{n-1}$}
{\int_{0}^{\infty} c_0{\times}e^{-\epsilon r}  r^{n-1} ~dr \int _{0}^{\pi} c_1 \sin^{n-2}\theta_1 ~d\theta_1 \dots  \int _{0}^{\pi} c_{n-2} \sin\theta_{n-2} ~d \theta_{n-2} \int_{0}^{2\pi} c_{n-1} ~d \theta_{n-1}
~.}
\end{Reason}


We recognise the form of this integral as the stochastic representation of a spherically symmetric distribution. The first component is the radial component and the remainder of the integrals represent the uniform distribution on the $n$-sphere. We can compute the constant $c_0$ by equating the radial component with a univariate distribution. That is,

\begin{Reason}
\Step{}
{\int_{0}^{\infty} c_0 {\times}e^{-\epsilon r}  r^{n-1} ~dr = 1}
\WideStepR{$\Implies$}{$c_0$ is constant}
{c_0 \int_{0}^{\infty} e^{-\epsilon r}  r^{n-1} ~dr = 1}
\WideStepR{$\Implies$}{Integration by parts}
{c_0 ( \frac{-1}{\epsilon} r^{n-1}e^{-\epsilon r} ) \Big|_0^{\infty} - c_0 \int_{0}^{\infty} \frac{-1}{\epsilon} e^{-\epsilon r}  (n-1) {\times} r^{n-2} ~dr = 1}
\WideStepR{$\Implies$}{$\lim_{r\to\infty} r^{n-1} e^{-\epsilon r} = 0$; simplifying}
{c_0 (n-1) / \epsilon \int_{0}^{\infty} e^{-\epsilon r} r^{n-2} ~dr = 1}
\WideStepR{$\Implies$}{Induction on $r$}
{c_0 (n-1)! / \epsilon^n = 1}
\WideStepR{$\Implies$}{Rearranging}
{c_0 = \epsilon^n / (n-1)!}
\end{Reason}

And now we can deduce the PDF of the radial distribution.


\begin{Reason}
\Step{}
{c_0 \times r^{n-1} e^{-\epsilon r}}
\WideStepR{$=$}{Using $c_0$ proven above}
{ r^{n-1} e^{-\epsilon r} {\times} \epsilon^n / (n-1)!}
\WideStepR{$=$}{\Def{d1732}}
{\Gam{n}{1/\epsilon}(r)}
\end{Reason}

\end{proof}


\begin{corollary}\label{c1234}
The n-dimensional Laplacian $\Lap{n}{\epsilon}({v}) = c_n^\epsilon \times e^{-\epsilon |\!|v|\!|}$ has constant $c_n^\epsilon$ given by
\[
          c_n^\epsilon = \epsilon^n / (n{-}1)! ~S_{n{-}1}(1) 
\]
where $S_{n{-}1}(1)$ is the surface area of the $n$-dimensional unit sphere.
\end{corollary}
\begin{proof} 
This follows from the observation that the integral
\[
  \int _{0}^{\pi} c_1 \sin^{n-2}\theta_1 ~d\theta_1 \dots  \int _{0}^{\pi} c_{n-2} \sin\theta_{n-2} ~d \theta_{n-2} \int_{0}^{2\pi} c_{n-1}  ~d \theta_{n-1}
\]
must sum to $1$ since we defined $c_0$ such that the radial integral was a probability distribution.
\end{proof}

\begin{lemma}\label{l1600}
For higher dimensions than $1$ the probability of a random vector being selected within a region is determined by a multiple integral; for the special case that the region is $D(R) \Defs \{v~|~ |\!|v|\!| \leq R \}$, then when $v$ is sampled from a $\Lap{n}{\epsilon}$ distribution, the probability that it is contained in $D(R)$, denoted ${\cal L}_n^\epsilon(R)$, is given by:
\begin{equation}
{\cal L}_n^\epsilon(R) \Wide{\Defs} 1{-} e^{-\epsilon R} \times e_{n-1}(\epsilon R)~,
\end{equation} 
where $e_k(\alpha) \Defs \sum_{0\leq i \leq k} \alpha^i/i!$~.
\end{lemma}

\begin{proof}
Using \Lem{l1328} we can calculate this probability using the radial (Gamma) distribution, since we defined the angular and radial distributions independently. That is,
\begin{align}\label{e3245}
    {\cal L}_n^\epsilon(R) = \int_0^R ~ r^{n{-}1} e^{-\epsilon r} {\times} \epsilon^n / (n{-}1)!~ dr
\end{align}
This has well-known CDF given by 
\[ 
   {\cal L}_n^\epsilon(R) = 1 - \sum_{k=0}^{n{-}1} \frac{(R \epsilon)^k}{k!} e^{-\epsilon R}
\]
which we now prove. We note that
\begin{align}\label{e3246}
   {\cal L}_1^\epsilon(R) = \int_0^R \epsilon e^{-\epsilon r} ~dr ~=~ 1 - e^{-\epsilon R}
\end{align}
and using integration by parts we see that
\begin{align}\label{e3247}
   {\cal L}_n^\epsilon(R) &= \epsilon^n / (n{-}1)! {\times} [ \frac{-r^{n{-}1}}{\epsilon} e^{-\epsilon r} ] \Big|_0^R + \int_0^R r^{n{-}2} e^{-\epsilon r} {\times} \epsilon^{n{-}1} / (n{-}2)! ~dr \nonumber \\
            &= {-}\frac{\epsilon^{n{-}1} R^{n{-}1}}{(n{-}1)!} e^{-\epsilon R} +  {\cal L}_{n{-}1}^\epsilon(R)
\end{align}
And therefore
\begin{Reason}
\Step{}
{\textit{Probability of outputting an element in} ~ D(R)}
\WideStepR{$=$}{\Eqn{e3245}}
{\int_0^R ~ r^{n{-}1} e^{-\epsilon r} {\times} \epsilon^n / (n{-}1)!~ dr}
\WideStepR{$=$}{\Eqn{e3246} and \Eqn{e3247}; induction}
{1-e^{-\epsilon R} [1 + R\epsilon + \frac{R^2 \epsilon^2}{2!} + \frac{R^3 \epsilon^3}{3!} + \dots + \frac{R^{n{-}1} \epsilon^{n{-}1}}{(n{-}1)!}]}
\WideStepR{$=$}{Arithmetic}
{1-e^{-\epsilon R}{\times} \sum_{k=0}^{n{-}1} \frac{R^k\epsilon^k}{k!}}
\WideStepR{$=$}{Simplifying}
{1-e^{-\epsilon R}{\times} e_{n{-}1}(\epsilon R)}
\end{Reason}
\end{proof}

We now present the proof of \Thm{l1534}. It follows as a consequence of of the next theorem.


\begin{theorem}
\[
{\int_{v \in Z_M}\prod_{1\leq i\leq N} \Lap{n}{\epsilon}(\underline{b}_i {-} v_i) dv_1\dots dv_N}  \Wide{\geq} 1- e^{-\epsilon N \Delta}e_{n{-}1}(\epsilon N\Delta)\frac{((c_n^\epsilon V_n(N\Delta))^N{-} 1)}{c_n^\epsilon V_n(N\Delta){-} 1}
\]

%
\begin{proof}
Let $\underline{b} \in (\Real^n)^N$ be a (fixed) vector representation of the bag $b$. For $v\in (\Real^n)^N$, let  $v^\circ \in {\mathbb B}\Real^n$ be the bag comprising the $N$ components if $v$.  Observe that $NE_{|\!|\cdot |\!|}(b, v^\circ) \leq  M_{|\!|\cdot |\!|}(\underline{b}, v)$, and so  
\begin{equation}\label{e1207}
Z_M = \{ v~ |~ M_{|\!|\cdot |\!|}(\underline{b}, v) \le N\Delta\} \Wide{\subseteq} \{ v ~|~ E_{|\!|\cdot |\!|}(b, v^\circ) \le \Delta\} = Z_E~.
\end{equation}
 Thus the probability of outputting an element of $Z$ is the same as the probability of outputting $Z_E$, and by \Eqn{e1207} that is at least the probability of outputting an element from $Z_M$ by applying a standard n-dimensional Laplace mechanism  to each of the components of $\underline{b}$.  We can now compute:
\begin{Reason}
\Step{}
{\textit{Probability of outputting an element in} ~ Z_E}
\WideStepR{$\geq$}{\Eqn{e1207}}
{\int_{v \in Z_M}\prod_{1\leq i\leq N} \Lap{n}{\epsilon}(\underline{b}_i {-} v_i) dv_1\dots dv_N} 
\WideStepR{$=$}{\Lem{l1328}}
{\int_{v \in Z_M}\prod_{1\leq i\leq N} c_n^\epsilon  e^{-\epsilon |\!| \underline{b}_i {-} v_i|\!|} dv_1\dots dv_N~.} 
\end{Reason}
The result follows by completing the integration, and applying simplifying approximations.


Let

\begin{equation}\label{e1548}
I_N \Wide{\Defs} {\int_{0 \leq \sum_{1 \leq j \leq N} |\!|v_j|\!| \leq R}~~ \prod_{1\leq i\leq N} \Lap{n}{\epsilon}(v_i)dv_1\dots dv_N}
\end{equation}

We rewrite RHS this as:

\begin{equation}\label{e1548-B}
{\int_{0 \leq \sum_{1 \leq j \leq N{-}1} |\!|v_j|\!| \leq R}~~ \prod_{1\leq i\leq N-1} \Lap{n}{\epsilon}(v_i)  \int_{0 \leq |\!|v_N|\!| \leq (R{-}  \sum_{1 \leq j \leq N{-}1}  |\!|v_j|\!|)}  \Lap{n}{\epsilon}(v_N) dv_N ~ dv_1\dots dv_{N{-1}}}
\end{equation}

Using \Lem{l1600} we can simplify the integral for $v_N$, to obtain:

\begin{equation}\label{e1548-C}
{\int_{0 \leq \sum_{1 \leq j \leq N{-}1} |\!|v_j|\!| \leq R}~~ \prod_{1\leq i\leq N-1} \Lap{n}{\epsilon}(v_i)  [ 1 - e^{-\epsilon(R{-}  \sum_{1 \leq j \leq N{-}1}  |\!|v_j|\!|)} e_{n{-}1}(\epsilon(R{-}  \sum_{1 \leq j \leq N{-}1}  |\!|v_j|\!|))  ] ~ dv_1\dots dv_{N{-1}}}
\end{equation}

Rewriting the product $\prod_{1\leq i\leq N-1} \Lap{n}{\epsilon}(v_i)$ as $(c_n^\epsilon)^{N{-}1} e^{-\epsilon(\sum_{1 \leq j \leq N{-}1}  |\!|v_j|\!|) }$ and using \Eqn{e1548} we can simplify \Eqn{e1548-C} to

\begin{equation}\label{e1548-D}
 I_{N{-}1} - e^{-\epsilon R} (c_n^\epsilon)^{N{-}1}{\int_{0 \leq \sum_{1 \leq j \leq N{-}1} |\!|v_j|\!| \leq R}~~    e_{n{-}1}(\epsilon(R{-}  \sum_{1 \leq j \leq N{-}1}  |\!|v_j|\!|)) ~ dv_1\dots dv_{N{-1}}}
\end{equation}

We make two approximations. Observe that $e_{n{-}1}$ is an increasing function, and that within the region of integration we have
\[
(\epsilon(R{-}  \sum_{1 \leq j \leq N{-}1}  |\!|v_j|\!|)) \Wide{\leq} \epsilon R~,
\]
so that \Eqn{e1548-D} is at least
\begin{equation}\label{e1548-E}
 I_{N{-}1} - e^{-\epsilon R} (c_n^\epsilon)^{N{-}1} e_{n{-}1}(\epsilon R){\int_{0 \leq \sum_{1 \leq j \leq N{-}1} |\!|v_j|\!| \leq R} 1 ~ dv_1\dots dv_{N{-1}}}
\end{equation}

The final simplification is to note that the integral is no more than $V_n(R)^{N{-}1}$, where $V_n(R)$ is the volume of an $n$-dimensional sphere. Putting all this together we have:

\begin{equation}\label{e1548-F}
I_N \Wide{\geq}  I_{N{-}1} - e^{-\epsilon R} (c_n^\epsilon)^{N{-}1} e_{n{-}1}(\epsilon R) V_n(R)^{N{-}1}
\end{equation}

We can now unwind this inequation to obtain the result, noting that $R \Defs N\Delta$.
\end{proof}
\end{theorem}

For the proof of  \Thm{l1534}, we need to make some further simplifications by using some additional constraints on the data.

In our application we know that the word2vec embeddings are typically of the order $n \geq 30$. In this case we can make the following approximations to $ c_n^\epsilon V_n(R)$.

From \Cor{c1234} we can compute $c_n^\epsilon = c_0 / S_{n-1}(1)$, where $c_0 = \epsilon^n/(n{-}1)!$ and $S_n(1)$ is the surface of an $n$-dimensional unit sphere. Using exact formulae for $V_n(R)$ and $S_{n-1}(1)$ we obtain:

\begin{equation}\label{e1128}
c_n^\epsilon V_n(R) \Wide{=} \epsilon^n/(n-1)! \times \frac{\pi^{n/2} R^n}{\Gamma(1+ n/2)} \times \frac{\Gamma(1 + n/2)}{n\pi^{n/2}} = \frac{(\epsilon R)^n}{n!}
\end{equation}

Using Stirling's approximation for $n!$ we obtain

\begin{equation}\label{e1128-B}
c_n^\epsilon V_n(R) \Wide{\approx}  \left( \frac{\epsilon e R}{n} \right)^n \times \frac{1}{\sqrt{2\pi n}}
\end{equation}

Comparing to our formula for  \Thm{l1534}, we can see that if we set $\epsilon e R \leq n$ then \Eqn{e1128-B} is less than $\frac{1}{\sqrt{2\pi n}}$ and so, for $n \geq 30$, we have 

\[
\frac{((c_n^\epsilon V_n(N\Delta))^N{-} 1)}{c_n^\epsilon V_n(N\Delta){-} 1} \approx 1
\]

giving finally that if $\epsilon e N\Delta \leq n$ then the utility calculation reduces to:

\[
{\int \dots \int_{v \in Z_M}\prod_{1\leq i\leq N} \Lap{n}{\epsilon}(\underline{b}_i {-} v_i) dv_1\dots dv_N}  \Wide{\geq} 1- e^{-\epsilon N \Delta}e_{n{-}1}(\epsilon N\Delta)
\]

which would be expected of a linear-like integration over an n-dimensional variable.

\section{Experimental Details}
\label{app:exper}

Here we describe further details of our constructed dataset and implemented inference mechanisms to support replicability, along with some additional analysis.

\subsection{Dataset Construction}
\label{app:data}

\paragraph{Document Sets}
Following \cite{koppel2011authorship}, we take the first 2000 words of each story to constitute the known-author \textsc{text}, and the final 1000 words of each to constitute the unknown-author \textsc{snippet}.  We then normalise the text (removing stopwords and punctuation, and lowercasing, which are standard for topic classification) and then --- because our mechanism requires each bag-of-words document to be of a fixed size $N$ --- truncate each text or snippet to the length of the shortest one in the set.  So for our 20-author set the texts are of length 420, and for the 50-author set length 402.

\paragraph{Word Embeddings}
The Google News word2vec embeddings\footnote{\url{https://code.google.com/archive/p/word2vec/}}  are 300-dimensional embeddings that were trained on about 100 billion words of news text and in full contain about 3 million words and phrases.  To make our experiments computationally feasible, we restricted our vocabulary to the 100,000 most frequent words for our 20-document dataset and the 30,000 most frequent words for our 50-document dataset.

\subsection{Inference Mechanism Implementation Details}
\label{app:impl}

\paragraph{Different-representation author identification}
The algorithm of \cite{koppel2011authorship} does not require any training for our purposes.  (In tasks where a ``don't know'' answer is permitted, there is a threshold parameter $\sigma$ that can be learnt, but we do not use this.)  There are a few hyperparameters to the method (e.g. size of character n-grams, number of character n-grams in the feature vector); for the most part we use the hyperparameter settings of the replication we used as our starting point\footnote{\url{https://github.com/pan-webis-de/koppel11}} \cite{potthast-etal:2016:ECIR}, which were set on the basis of the empirical analysis of \cite{koppel2011authorship}.  Only the minimum text length for training is changed to 400, given the length of our texts.

\paragraph{Different-representation topic classification}
In training, we use pretrained embeddings trained on about 16 billion words of English Wikipedia text.\footnote{\url{https://fasttext.cc/docs/en/english-vectors.html}}  Our in-domain dataset consists of texts from 500 authors chosen at random, some of whom had written multiple stories; from this we derived a training set of 448 known-author texts of size 2000 each, split evenly between the two topics, and a comparable validation set of 111 known-author texts.  We trained the classifier for 25 epochs, with learning rate 1.0. On the validation set, the accuracy is 0.937.

\subsection{Analysis of Examples}
\label{app:analysis}

\begin{table}[]
    \centering
    
    {\footnotesize
    \begin{tabular}{|l||r|r|l|}
    \hline
snippet & $i$ & $d_i$ & $a_i$\\
\hline
\hline
4-Bad-Boys-of-Twilight-0.txt 
& 1 & 2.112 &	AliceInMyWonderland\\
& 2 & 2.136	& AliciaMarieSwan\\
& 3 & 2.147 &	Bad-Boys-of-Twilight\\
\hline
4-Anything-Goes-Twific-Contest-1.txt 
& 1 &	1.618 &	Anything-Goes-Twific-Contest\\
& 2 &	1.773 &	91BlackMoon\\
& 3 &	1.805 &	AliciaMarieSwan\\
\hline
0-dorahatesexploring-0.txt 
& 1 &	1.844 &	dorahatesexploring\\
& 2 &	1.879 &	Dr-Mini-me\\
& 3 &	1.974 &	k-kizkhalifa\\
\hline
4-chiriko1117-0.txt 
& 1 &	1.913 &	Dr-Mini-me\\
& 2 &	1.948 &	AliciaMarieSwan\\
& 3 &	1.986 &	91BlackMoon\\
\hline
0-FateRogue-2.txt 
& 1 &	1.798 &	FateRogue\\
& 2 &	1.854 &	Dr-Mini-me\\
& 3 &	1.896 &	AliciaMarieSwan\\
\hline
    \end{tabular}
    
    \vspace{0.25cm}
    
    }
    \caption{Distances for a sample of unmodified snippets from the 20-author dataset. $d_i$ is the distance to $i$th closest neighbour known-author text; $a_i$ is the identity of the $i$th closest author.}
    \label{tab:distances}
\end{table}

Table~\ref{tab:distances} illustrates the three closest distances for a sample of unknown-author snippets, along with their authors.  It can be seen that the distances are relatively close, a consequence of the high-dimensional Euclidean space.

\begin{table}[]
    \centering
    {\footnotesize
    \begin{tabular}{|l||l|l|l|l|l|l|}
    \hline
$\epsilon$ & \multicolumn{6}{c|}{words 1--6}\\
    \hline
    \hline
none & heard & grandfather & pollux & dubbed & uncle & alphard\\
\hline
$30$ & heard & grandfather & pollux & orient & grandma & alphard\\
$20$ & Walt\_Disney & grandfather & pollux & orient & grandma & alphard\\
$10$ & Walt\_Disney & Guinness\_Book & pollux & tilted & Public\_Defender & alphard\\
\hline
    \end{tabular}
    
    \vspace{0.1cm}
    
        \begin{tabular}{|l||l|l|l|l|l|}
    \hline
$\epsilon$ & \multicolumn{5}{c|}{words 7--11}\\
    \hline
    \hline
none & sympathizer & disowned & regulus & hurt & felt\\
\hline
$30$ & sympathizer & disowned & regulus & Harrison & feels\\
$20$ & sympathizer & Records & regulus & culprits & algebra\\
$10$ & premeditated\_murder & ENERGY\_STAR & regulus & culprits & plagiarism\\
\hline
    \end{tabular}
    
    \vspace{0.25cm}
    
    }
    \caption{The beginnings of one particular snippet under various levels of $\epsilon$.}
    \label{tab:samples}
\end{table}

Table~\ref{tab:samples} illustrates some of the changes introduced by the privacy mechanism under various levels of $\epsilon$ for a single sample snippet.  $\epsilon=30$ produces very minor changes, mostly semantically close (\textit{felt} $\rightarrow$ \textit{feels},
\textit{uncle} $\rightarrow$ \textit{grandma}), but with some greater randomness as well (\textit{hurt} $\rightarrow$ \textit{Harrison}).  This increases as $\epsilon$ decreases, until there are some very unlikely words (e.g. \textit{ENERGY\_STAR}), as expected.

In a practical application, such unlikely words or phrases like \textit{ENERGY\_STAR} would look rather out of place with respect to the domain.  However, our choice of word2vec vocabulary for this experiment was in a sense arbitrary; a practical application could use a vocabulary tailored to the domain, by (say) culling entries that do not appear in a training set.

\end{document}